\documentclass[12pt, draftclsnofoot, onecolumn]{IEEEtran}

\usepackage{mathptmx} 
\usepackage{amsmath,amssymb} 
\usepackage{dsfont}
\usepackage{verbatim}
\usepackage{multicol,longtable}
\usepackage{epsfig}
\usepackage{txfonts}
\usepackage{calc}
\usepackage{enumerate}
\usepackage{bm}
\usepackage{epstopdf}
\usepackage{subfigure}

\usepackage[ruled,vlined,linesnumbered]{algorithm2e}
\usepackage{algorithmic}

\newtheorem{thm}{Theorem}

\newtheorem{remark}{Remark}

\newcommand{\EE}{\mathbb{E}}

\def\1{{\mathbf 1}}  

\begin{document}

\title{Cross-Layer Design of Wireless Multihop Networks over Stochastic Channels with Time-Varying Statistics}


\author{\IEEEauthorblockN{Eleni~Stai$^{1}$, Michail~Loulakis$^{2}$,
Symeon~Papavassiliou$^{1}$}\\ \IEEEauthorblockA{$^{1}$ School of
Electrical \& Computer Engineering\\ $^{2}$ School of Applied
Mathematical \& Physical Sciences\\ National Technical University of
Athens, Athens, Zografou, 15780, Greece.\\Emails:
estai@netmode.ntua.gr, loulakis@math.ntua.gr,
papavass@mail.ntua.gr}}


%

\IEEEcompsoctitleabstractindextext{
\begin{abstract}
Network Utility Maximization is often applied for the cross-layer
design of wireless networks considering known wireless channels.
However, realistic wireless channel capacities are stochastic
bearing time-varying statistics, necessitating the redesign and
solution of NUM problems to capture such effects. Based on NUM
theory we develop a framework for scheduling, routing, congestion
and power control in wireless multihop networks that considers
stochastic Long or Short Term Fading wireless channels.
Specifically, the wireless channel is modeled via stochastic
differential equations alleviating several assumptions that exist in
state-of-the-art channel modeling within the NUM framework such as
the finite number of states or the stationarity. Our consideration
of wireless channel modeling leads to a NUM problem formulation that
accommodates non-convex and time-varying utilities. We consider both
cases of non orthogonal and orthogonal access of users to the
medium. In the first case, scheduling is performed via power
control, while the latter separates scheduling and power control and
the role of power control is to further increase users' optimal
utility by exploiting random reductions of the stochastic channel
power loss while also considering energy efficiency. Finally,
numerical results evaluate the performance and operation of the
proposed approach and study the impact of several involved
parameters on convergence.
\end{abstract}

}

\maketitle \IEEEdisplaynotcompsoctitleabstractindextext
\IEEEpeerreviewmaketitle

\section{Introduction}
\label{sec:intro} Network Utility Maximization (NUM) is a very
popular tool in the communications research community, for
cross-layer design and optimization of wireless networks. Typically,
a utility function is assigned to each network flow
(source-destination pair), and the sum of all utilities over the
network is maximized, subject to network stability constraints. Most
approaches in literature applying NUM, consider ideal or stationary
and ergodic wireless channels. However, under realistic conditions,
fading occurring in wireless channels hampers the performance in
wireless communications, leading to stochastic, i.e. time and space
varying, random (thus unknown), wireless link capacities possibly
bearing time-varying statistics. Therefore, it becomes necessary to
reformulate/redesign and solve the basic NUM problem for
incorporating realistic stochastic wireless channel conditions by
addressing non-stationarity issues, and also transient phenomena
occurring when the network operates for a finite time interval.

This paper aims at treating this problem by proposing a novel
optimization framework for joint congestion control, routing,
scheduling and power control based on NUM, under stochastic possibly
non-stationary Long Term Fading (LTF) or Short Term Fading (STF)
wireless channels. Congestion control determines the optimal
sources' data production rates, routing determines the optimal
routes of the flows within the network, while the set of
transmitting links and their corresponding transmission powers are
chosen based on scheduling and power control. The proposed
optimization framework refers to a finite network operation and it
is developed for two cases, where the first one deals with
non-orthogonal access to the wireless medium and the second deals
with orthogonal access to the wireless medium. In the first case,
scheduling is performed via power control, while in the latter case,
scheduling and power control are separated. On the one hand,
non-orthogonal access leads to a very hard to solve cross terminal
power control/scheduling problem in the physical layer. On the other
hand, orthogonal access to the medium allows for more efficient and
distributed power control in the physical layer, although
introducing the need of the NP-hard computation of all the
independent sets of the network graph for scheduling. Importantly,
in the case of orthogonal access, the role of power control is to
further boost link capacity and consequently source rates by
exploiting good channel states and to save energy when the channel
state is destructive for the transmitted signal.


The structure of the rest of the paper is as follows. Section
\ref{sec:relatedwork} presents the related literature while
summarizing the basic contributions of this paper. Then, Section
\ref{sec:channelmodels} describes the considered STF and LTF channel
models derived with the use of stochastic differential equations,
while Section \ref{sec:sysmodel} presents the considered system
model. Sections \ref{sec:nonortho} and \ref{sec:orthogonal} focus on
the analysis and solution of the proposed optimization framework in
the cases of non-orthogonal and orthogonal access to the wireless
medium respectively. Finally, in Section \ref{sec:simulation}
numerical results are presented to evaluate the proposed approach
and Section \ref{sec:conclusions} concludes the paper.

\section{Related Work \& Contributions}\label{sec:relatedwork}Several works exist in the literature targeting at
incorporating the stochastic wireless channels in the NUM problem's
formulation and solution. In \cite{chiang05}, the channel quality is
expressed via the SIR (Signal-to-Interference-Ratio), since the
latter is affected by interference due to parallel transmissions and
LTF. It is assumed that the LTF parameters are deterministic and
slowly varying, allowing for the algorithms which perform joint
congestion and power control to converge in the meantime of their
change. 
In \cite{papandriopoulos08}, \cite{fischione11}, the case of
composite fading (LTF and STF) is examined, considering channel
conditions that vary faster than the algorithm's convergence, via
the use of outage-probabilities in the NUM problem formulation. In
\cite{papandriopoulos08}, \cite{fischione11}, STF follows Rayleigh
and Nakagami distributions respectively, where however the
statistics of the distributions remain invariant. A different
approach is followed in \cite{neill09, firouzabadi10}, where NUM is
extended to Wireless NUM (WNUM), with random channel conditions.
WNUM leads to policies for controlling the network by responding
optimally to the change of the channel state, based on random
samples without a priori knowledge of the wireless channels'
statistics. Although the statistics of the channel may be unknown,
the latter is considered as stationary and ergodic.

In the sequel, in \cite{marques}, NUM is employed to perform, joint
congestion control, power control, routing and scheduling assuming
that the channel fading process is stationary and ergodic. Similarly
in \cite{chen06,huang13}, joint congestion control, routing and
scheduling is performed in the framework of NUM while assuming that
there is a finite number of channel states. In these works, the
network functions in time slots, where during each time slot the
channel state remains stable and changes randomly and independently
on the boundary of time slots. Finally, in \cite{chen11} the
convergence of primal-dual algorithms for solving NUM is studied
under wireless fading channels with time-varying parameters (and
thus statistics). Time-varying statistics of wireless channels lead
to time-varying optimal solutions of the NUM problem necessitating
the study of how well the solution algorithms track the changes in
the optimal values. However, it is assumed that the channel fading
parameters vary following a Finite State Markov Chain.

In a nutshell, in the existing body of research work in literature,
the wireless channel modeling in the framework of NUM is
characterized by one or combinations of the following assumptions:
(a) The channel process is stationary and ergodic. (b) The
statistics of the wireless channel are fixed in time (and known), or
vary at slower rate than the one of the network control algorithms'
convergence. (c) The statistics of the channel change according to a
Finite State Markov Chain. All previous approaches are not capable
of capturing and tracking complex time and space variations in the
propagation environment of realistic systems \cite{olama10}. In
\cite{olama10}, wireless channel models for both LTF and STF are
introduced based on Stochastic Differential Equations (SDEs)
\cite{oksendal03} in order to capture higher order dynamics of the
wireless channel. In this case the wireless channel is modeled via
stochastic processes which may have time-varying statistics. By
means of SDEs, it is possible to express an LTF or an STF channel
capturing both space and time variations \cite{olama10},
\cite{menemenlis99}, as it will be described in detail in Section
\ref{sec:channelmodels}.

In our paper, the NUM problem is reformulated and solved using the
SDE model to capture the wireless channel state. Emphasis is placed
on LTF especially for demonstration purposes. Due to the possible
non-stationarity of the wireless channel we cannot formulate the NUM
problem based on the stationary mean values of the involved
optimization variables (e.g. \cite{marques}, \cite{chen06},
\cite{stai14}). On the contrary we will adopt the stochastic optimal
control problem's formulation \cite{fleming06} based on expected
values over time integrals, thus also allowing for the consideration
of a finite time duration of the network's operation. In
\cite{stai15}, we proposed a preliminary version of this approach
focused on congestion and power control in the case of orthogonal
access to the wireless medium. The basic contributions of this paper
can be summarized as:
\begin{itemize}
\item We develop a framework for the cross-layer design and control of the operation of wireless multihop networks, i.e. for congestion and power control, routing and scheduling,
 over wireless channels (LTF or STF) that are stochastic but not necessarily
stationary.
\item The proposed problem formulation
adopts a (more realistic) finite duration of the network's operation
(e.g. corresponding to the case of finite battery levels of the
wireless nodes or of finite flows) where the wireless channel may
still operate in a transient state, even though there may exist a
limiting stationary distribution. The extension to an infinite
duration of the network's operation is discussed.
\item Utilities are not necessarily convex functions but adopt a more general form (more specifically a continuous differentiable one), contrary to the related papers in
literature (e.g. \cite{marques}, \cite{chen06}, \cite{stai14}) that
assume convex utility forms. This fact is important since it allows
for addressing the case of real time traffic modeled, for example,
by sigmoidal utilities \cite{tychogiorgos13}. Zero duality gap is
analytically proven in this case of general utilities, following a
technique that leverages from the wireless channels' continuous
stochastic modeling.
\item The proposed problem formulation (expected values over time integrals) allows for the adoption of time-varying utility functions. This serves the purpose of evolving users'
preferences/needs and is also aligned with the finite network
duration, i.e. nodes may desire to produce significantly less data
close to the end of the network operation. To the best of our
knowledge, this fact has not been addressed in the literature.
\item Power control is shown to further boost link capacity
and thus source rates by exploiting good channel conditions while
saving energy in case of destructive channel conditions.
\item Finally, we prove an interesting theorem in the case of LTF, elucidating a basic advantage with
respect to the optimal users' utilities when exploiting the random
channel fluctuations, compared to the conventional NUM problem
formulation (e.g. \cite{chen06}, \cite{stai14}). Specifically, it is
proven that, contrary to what is possibly expected, a higher value
of the diffusion coefficient of the wireless channels' power loss
leads to higher optimal sum of users' utilities, fact that cannot be
captured by the conventional NUM modeling approach. We also show via
numerical evaluations that a higher diffusion coefficient achieves
simultaneously a reduced power consumption leading to energy
efficiency. This result emphasizes the importance of utilizing a
more realistic power loss model such as the one of an SDE, as
opposed to the mean power loss model used in the conventional NUM
problem formulation.
\end{itemize}
\section{Background on Wireless Channel Modeling via SDEs} \label{sec:channelmodels}
The objective of this section is to briefly describe the SDE-based
LTF and STF channel models, developed in the literature, and their
assumptions, upon which, the optimization problems of the following
sections will be formulated and solved.
\subsection{Long Term Fading (LTF)}\label{sec:ltf}
LTF consists of path loss and shadowing \cite{goldsmith05}. Path
loss is due to the dissipation of the transmitted power and the
effects of the propagation channel, while shadowing is caused by
obstacles between the transmitter and the receiver. LTF depends on
the geographical area and occurs in sparsely populated or suburban
areas. 
Before describing the dynamic in time LTF model for the wireless
channels \cite{olama10}, we recall the conventional LTF model, where
the power loss $PL$ along a given link $(i,j)$ between nodes $i$ and
$j$ in Euclidean distance $d_{ij}$, is given \cite{goldsmith05}:
\begin{equation}
PL(d_{ij})[dB]=\overline{PL}(d_{0})[dB]+ 10\gamma \log
\big(\frac{d_{ij}}{d_{0}} \big)+\tilde{Z}, ~d_{ij}\geq d_0,
\label{eq:pathloss}
\end{equation}
\noindent where $\gamma$ is the power loss exponent and depends on
the wireless propagation medium, $d_{0}$ is the reference distance,
$\overline{PL}(d_{0})$ is the expected power loss on the reference
distance, and $\tilde{Z}\sim \mathcal{N}(0;\sigma^2)$, is a gaussian
random variable with zero mean and variance $\sigma^2$, used to
model any uncertainty in the propagation
environment. 
Note that the statistics (mean (denoted as $\overline{PL}(d_{ij})$)
and variance) of the conventional LTF model are invariant in time.

In the following, we describe the extension of the LTF model to
dynamically changing conditions in time, as it is developed in
\cite{olama10}. Specifically, the random variable, $PL(d_{ij})[dB]$,
of Eq. (\ref{eq:pathloss}), becomes a stochastic process denoted as
$\{X_{ij}(t)\}_{t\geq t_0}$ ([dB]), where $t$ represents time. Time
dependence is used to capture time variations of the propagation
environment due to e.g. movement of objects and scatterers in the
area surrounding the network. In a similar spirit with Eq.
(\ref{eq:pathloss}), $\{X_{ij}(t)\}_{t\geq t_0}$, represents the
power lost by the signal during a transmission from $i$ to $j$ at a
particular distance $d_{ij}$. Although, $\{X_{ij}(t)\}_{t\geq t_0}$
depends on the distance $d_{ij}$, we do not explicitly model this
dependence
as the network considered is static.  

In \cite{olama10}, $\{X_{ij}(t)\}_{t\geq t_0}$, $\forall (i,j)$, are
modeled as solutions of mean reverting linear SDEs, given as:
\begin{eqnarray}
dX_{ij}(t)=\beta_{ij}(t)\big(\gamma_{ij}(t)-X_{ij}(t)\big)dt+\delta_{ij}(t)dW_{ij}(t),~X_{ij}(t_0)\sim
\mathcal{N}(\overline{PL}(d_{ij})[dB];\sigma^2),
\label{eq:sdepathloss}
\end{eqnarray}\noindent where $\{W_{ij}(t)\}_{t\geq t_0}$, $\forall (i,j)$, are
independent standard Brownian motions defined over a filtered
probability space $(\Omega,\mathcal{F},\{\mathcal{F}_t\}_{t\geq
t_0},\mathbf{P})$ and each one being independent of the
corresponding $X_{ij}(t_0)$. 
$\{\mathcal{F}_t\}_{t\geq t_0}$ is the filtration produced by
$X_{ij}(t_0)$, $\forall (i,j)$, and the Brownian motions themselves.
For each $(i,j)$, $\gamma_{ij}(t)$ is the power loss level
$X_{ij}(t)$ is attracted to, $\beta_{ij}(t)$ is the positive speed
of this adjustment and finally, $\delta_{ij}(t)$ is the diffusion
coefficient of the SDE, determining the ``noise" of the channel. The
parameters $\beta_{ij}(t),\gamma_{ij}(t),\delta_{ij}(t)$, $\forall
(i,j)$, are assumed to be deterministic and can be estimated
directly from signal measurements following the approaches in
\cite{olama06}, \cite{olama09}, \cite{charalambous08}, which can be
implemented online, i.e. while receiving the signal measurements.
The existence of a strong solution to the SDE (\ref{eq:sdepathloss})
is satisfied if the relation $\int\limits_s^T \left\{\beta_{ij}(t)\left|\gamma_{ij}(t)\right|+\delta _{ij}^2(t)\right\}dt < \infty , \forall (i,j)$ holds \cite{oksendal03}. 
The time dependent attenuation coefficient (in squared magnitude)
equals to: 
$a_{ij}(t)=e^{-\frac{\ln10}{10}X_{ij}(t)}=e^{K X_{ij}(t)},~\forall
(i,j),~K=-\frac{\ln10}{10}$. 
%
%
%
In \cite{olama10}, it is shown that when all the parameters of the
SDE (\ref{eq:sdepathloss}) are time independent, its solution tends
to the conventional LTF model (Eq. (\ref{eq:pathloss})) as
$t\rightarrow \infty$ (which is stationary). In general when the
parameters of the SDE (\ref{eq:sdepathloss}) change with time,
$X_{ij}(t)$ is gaussian with time-varying statistics and a
stationary
distribution may not exist. 
\subsection{Short Term Fading (STF)}\label{sec:stf}
In a similar spirit as LTF, in \cite{menemenlis99}, \cite{olama09},
\cite{olama10} a stochastic model for STF wireless channels has been
developed, alleviating the assumption of stationarity. This kind of
signal fading is due to the constructive and destructive addition of
multipath components \cite{goldsmith05} created from reflections,
diffractions and scattering and usually occurring in densely built-up areas. 
The statistics of the STF models usually applied in the literature
(e.g. Rayleigh, Nakagami, Ricean, etc. \cite{papandriopoulos08},
\cite{goldsmith05}), are assumed constant over local areas (i.e. at
a microscopic level) \cite{menemenlis99}. However, STF wireless
channels are of stochastic nature with time varying statistics
mainly due to the continuous and arbitrary change of the propagation
environment if the transmitter, the receiver or objects between them
move. The latter is the main reason why in this paper, we adopt a
stochastic process with time varying statistics for modeling STF
channels.

For the models developed in \cite{menemenlis99}, \cite{olama09}, the
inphase and quadrature components of the wireless fading channels
are assumed conditionally uncorrelated gaussian random variables
(thus conditionally independent). In the case of flat fading, the
multipath components are not resolvable and can be considered as a
single path. Then, the inphase, $I$, and quadrature, $Q$, components
over one link (e.g. $(i,j)$) can be realized as:
\begin{eqnarray}
dX_{I}(t)=A_I(t) X_{I}(t)dt+B_I(t) dW_I(t),~X_{I}(t_0),~
I(t)=C_I X_{I}(t),\nonumber\\
dX_{Q}(t)=A_Q(t) X_{Q}(t)dt+B_Q(t) dW_Q(t),~X_{Q}(t_0),~ Q(t)=C_Q
X_{Q}(t),\nonumber \label{eq:sdestf}
\end{eqnarray}\noindent where $X_{I}(t),~X_{Q}(t)$ are the state
vectors of the inphase and quadrature components and
$\{W_I(t)\}_{t\geq t_0}$, $\{W_Q(t)\}_{t\geq t_0}$ are independent
standard Brownian motions corresponding to the inphase and the
quadrature components respectively, defined over a filtered
probability space $(\Omega,\mathcal{F},\{\mathcal{F}_t\}_{t\geq
t_0},\mathbf{P})$. The same model describes every link with
different parameter values $A_I(t),~A_Q(t)$,~$B_I(t)$, $B_Q(t)$,
$C_I,~C_Q$ and (independent) Brownian motions $W_I(t),~W_Q(t)$, but
this fact is not explicitly modeled for ease of presentation.
 The attenuation coefficient (in squared magnitude) is \cite{charalambous08} $a_{ij}(t)=I(t)^2+Q(t)^2.$ 
As in the case of LTF, the coefficients
$A_I(t),~A_Q(t),~B_I(t),~B_Q(t),~C_I,~C_Q$ can be obtained directly
via signal measurements following the methodology proposed in
\cite{charalambous08}, \cite{olama09} using the EM algorithm
together with Kalman filtering. This model leads to time-varying
mean and variance for the inphase and the quadrature components and
thus for the STF wireless channel and includes the Ricean, Rayleigh
and Nakagami distributions as special cases \cite{menemenlis99}.

In the rest of the paper, the vectors $X(t),~W(t)$ denote
collectively (for all links) the channel states and the Brownian
motions respectively at time $t$. Note that we assume that the
wireless channels are uncorrelated. This assumption is also made in
\cite{neely10}, where it is argued that inter-link correlations do
not impact the network capacity region and therefore the maximum
utility.
\section{System Model \& Assumptions}\label{sec:sysmodel}
We consider a static wireless multihop network with $N$ nodes and
$E$ directed links forming the set $\mathcal{E}$. The network serves
$F$ overlaying flows (source-destination pairs) over a finite
duration (lifetime) $[s,T]$. At time $t \in [s,T]$, $\lambda_i^d(t)$
data (e.g. packets) are produced from the source node $i$ for its
destination node $d$. Let $S_r(d)$ be the set of sources for node
$d$. Then, $\lambda_i^d(t)=0,~\forall i \notin S_r(d),~\forall t \in
[s,T]$. We denote with $r_{ij}^d(t)$ the communication traffic on
the link $(i,j)$ for destination $d$ at time $t\in [s,T]$. Then,
$R(t),~\Lambda(t),$ denote collectively the variables
$\{r_{ij}^d(t)\}_{\forall d,(ij)},~\{\lambda_i^d(t)\}_{\forall
i,d},$ respectively, at time $t$. The set $\mathcal{R}(i,d)$
consists of the one-hop (out-)neighbors of node $i$ which are
allowed to serve as next-hop nodes towards $d$ according to the
routing protocol under consideration. If there are no routing
constraints, we consider $\mathcal{R}(i,d)=\mathcal{N}^{out}(i)$,
where $\mathcal{N}^{out}(i)=\{j|(i,j)\in \mathcal{E}\}$. Also,
$r_{ij}^d(t)=0$, if $j \notin \mathcal{R}(i,d)$ and
$r_{ii}^d(t)=0,~r_{di}^d(t)=0,~\forall i,d,t\in [s,T]$. Furthermore,
we assume that the transmitter of the link $(i,j)$, i.e. the node
$i$, transmits with power $P_{ij}(t)$ at time $t \in [s,T]$. $P(t)$
expresses collectively the transmission powers of all links at time
$t \in [s,T]$.

Each source node associates its satisfaction for its produced data
for destination $d$, $\lambda_i^d$, at time $t \in [s,T]$, with a
time-varying continuous differentiable utility function
$U_i^d(\lambda_i^d, t)$. Several utility functions used in
literature belong in this category, such as the strictly convex and
increasing $a$-fair utility, including the logarithmic one
\cite{eryilmaz06}. It is further assumed that $U_i^d(\lambda_i^d,
t)$ is increasing with $\lambda_i^d$ and uniformly bounded as
$t\rightarrow \infty$. Also, a (continuous differentiable) cost
function, $J_{ij}(P_{ij})$, is assigned to each directed link
$(i,j)$ with respect to its transmission power $P_{ij}(t)$, $t \in
[s,T]$. In literature, $J_{ij}(P_{ij})$ is often assumed to be a
strictly convex function \cite{marques}.

The proposed cross-layer framework includes routing, scheduling,
power and congestion control. Routing (network layer) determines the
amount of traffic for each destination that will be served by every
link, by optimizing $R(t)$, $\forall t \in [s,T]$. Scheduling and
power control (MAC and physical layers) determine which links are
going to transmit and their transmission power by optimizing $P(t)$,
$\forall t \in [s,T]$. Finally, congestion control (transport layer)
optimizes $\Lambda(t)$, $\forall t \in [s,T]$. Therefore, the
proposed cross-layer scheme aims at determining the optimal values
of the control variables $R(t),~\Lambda(t),~P(t)$, $\forall t \in
[s,T]$, according to an optimality criterion designed with the aid
of the utility and cost functions defined above, considering the
channel models of Section \ref{sec:channelmodels}.

 Due to the considered underlying
channel processes, the control variables $R,~\Lambda,~P$ should be
in general defined as stochastic processes. Let us define the value
range for each $\lambda_i^d$, $U_{\lambda}=[0,\lambda_{\max}]$, and
the corresponding feasible set
$\mathcal{U}_{\lambda}=\left\{\lambda:[s,T]\times \Omega\rightarrow
U_{\lambda}: \lambda \mathrm{~is~} \{\mathcal{F}_t\}_{t\geq s}
\mathrm{~adapted} \right\}$. Then, $\Lambda \in
\mathcal{U}_{\lambda}^F$. Similarly, we define the value range for
each $r_{ij}^d$, $U_{r}=[0,R_{\max}]$, and the corresponding
feasible set $\mathcal{U}_{r}=\left\{r:[s,T]\times \Omega\rightarrow
U_{r}: r \mathrm{~is~} \{\mathcal{F}_t\}_{t\geq s} \mathrm{~adapted}
\right\}$. Then, $R \in \mathcal{U}_{r}^{E\times (N-1)}$. Finally,
we define the value range for each $P_{ij}$, $U_{P}=[0,P_{\max}]$,
and the corresponding feasible set
$\mathcal{U}_{P}=\left\{P:[s,T]\times \Omega\rightarrow U_{P}: P
\mathrm{~is~} \{\mathcal{F}_t\}_{t\geq s} \mathrm{~adapted}
\right\}$. Then, $P \in \mathcal{U}_{P}^{E}$. We will use
$\EE_{s,x}$ to denote expectations given the initial condition
$X(s)=x$.

At this point, we distinguish two cases with respect to the access
to the wireless medium. Based on the two types of access, two
cross-layer problems are developed. The first case concerns
non-orthogonal access to the wireless medium, in which the
transmitters are allowed to access the wireless medium
simultaneously (one frequency carrier is assumed), while the
interfering transmissions are considered as noise. For this case we
define the Signal-to-Interference-plus-Noise-Ratio $(SINR)$ for the
link $(i,j)$ as follows: 
$SINR_{ij}(t)=\frac{a_{ij}(t)P_{ij}(t)}{N_0+\sum_{(k,l)\in
\mathcal{I}_{ij}}a_{kj} P_{kl}(t)}$,
where $\mathcal{I}_{ij}$ denotes the subset of $\mathcal{E}$
containing the links that interfere with the link $(i,j)$. $N_0$
(Watts) stands for the average background noise at the receiver's
($j$) side and $a_{ij}(t)$ is defined in Section
\ref{sec:channelmodels} depending on the fading type. The capacity
of the link $(i,j)$ is given by the Shannon's formula in
bits/sec as 
$C_{ij}(P(t))=B_{ij}\log_2(1+SINR_{ij}(t))$,
where $B_{ij}$ ($Hertz$) is the wireless channel's bandwidth at link
$(i,j)$. The second case refers to the orthogonal access to the
wireless medium, where only non-interfering links can access
simultaneously the wireless medium. Orthogonal access decreases the
complexity of the proposed framework's operation in the physical
layer (power control), as it will be shown in later sections, while
it is nearly optimal when interference is strong \cite{marques}. In
this case, the connectivity graph of the wireless multihop network
is important in identifying the feasible schedules. Based on the
latter, the finite set of all possible independent sets of links
(i.e. links that do not interfere with each other) is constructed.
Only links belonging to the same independent set can access the
wireless medium simultaneously. In this case, the capacity of the
link $(i,j)$ is a concave function of $P_{ij}$, given in bits/sec as
$C_{ij}(P_{ij}(t))=B_{ij}\log_2\left(1+\frac{a_{ij}(t)P_{ij}(t)}{N_0}\right)$.

\section{Non-orthogonal Access to the Medium}
\label{sec:nonortho}
\subsection{Problem Formulation \& Analysis} \label{sec:Problform}
According to the discussion in Section \ref{sec:sysmodel}, the
optimization framework, denoted as $\mathbf{P_1}$, is formulated as
follows.
\begin{eqnarray}
P_1:=\max_{\Lambda \in \mathcal{U}_{\lambda}^F, ~R \in
\mathcal{U}_{r}^{E\times (N-1)},~ P \in \mathcal{U}_{P}^{E}}
\EE_{s,x}\left[\int_s^T \left(\sum_{i,d:i\in S_r(d)}
U_i^d(\lambda_i^d(t),t)-
\sum_{(i,j)}J_{ij}(P_{ij}(t))\right)dt\right]\nonumber \\
s.t.\nonumber\\
\EE_{s,x}\left[\int_s^T \lambda_i^d(t)dt +\int_s^T \sum_{j:i \in
\mathcal{R}(j,d)} r_{ji}^d(t)dt\right]\leq \EE_{s,x}\left[\int_s^T
\sum_{j\in \mathcal{R}(i,d)} r_{ij}^d(t)dt\right],~ \forall i,d \label{const:1}\\
\EE_{s,x}\left[\int_s^T \sum_{d} r_{ij}^d(t)dt\right]\leq
\EE_{s,x}\left[\int_s^T
C_{ij}(P(t))dt\right], ~\forall (i,j)\in \mathcal{E}\label{const:2}\\
\EE_{s,x}\left[\int_s^T\sum_{j\in
\mathcal{N}^{out}_i}P_{ij}(t)dt\right]\leq P_{i,\max}, ~\forall
i\label{const:3}
\end{eqnarray}
The objective function expresses the trade-off between the
accumulated for all sources utilities/satisfaction of producing data
and the accumulated cost due to power consumption for link
transmissions. Thus, its maximization targets at improving energy
efficiency by maximizing source rates while penalizing the cost of
power consumption for achieving them when $J_{ij}\neq 0,~\forall
(i,j)$, i.e. when power control is applied. The first constraint
relates to the flow conservation at each node and for each
destination, used to ensure high throughput for the examined time
interval $[s,T]$. The second constraint relates to the capacity
restriction (right side) due to power, channel and interference
limitations for each link. Finally, the third constraint, relates to
a limitation on the total power consumption of each node for the
examined time interval, $[s,T]$, (left side) according to its energy
storage, denoted as $P_{i,\max}>0$ (right side).

It is important to note that we consider continuous time network
operation (thus, using time integrals), for ease of presentation,
due to the continuous time evolution of the channel state (Section
\ref{sec:channelmodels}). The case of discrete time network
operation can be obtained trivially by replacing the integrals
$\int_s^T$ by sums $\sum_{t=0}^{t=N_L}$ where $N_L$ is the number of
time slots in $[s,T]$ considered for the network operation and each
denoted by $t$. In this case the channel state will be sampled at
each time slot as described in the subsequent sections. Furthermore,
for considering an infinite $T$, we should also divide every
integral by $T$, thus considering time averages.

This problem is non-convex due to the forms of the capacity, utility
and cost functions. Even if the utility and cost functions were
concave and convex respectively as commonly assumed in literature,
the problem would still be non-convex due to the capacity function
forms. However, we will prove that its duality gap is zero, which is
an important fact as it renders the Lagrange (dual)-based
optimization method optimal. The latter allows for devising
efficient algorithmic solutions as it leads to a separable
optimization problem with respect to the variables of each layer
while using the Lagrange multipliers for the communication between
adjacent layers for achieving a cross-layer optimal solution. Let
$\mu_i^d\geq 0, ~\forall i,d, ~l_{ij}\geq 0,~\forall (i,j)\in
\mathcal{E},~\nu_i\geq 0, ~\forall_i$ be the Lagrange multipliers
associated with the constraints (\ref{const:1}), (\ref{const:2}),
(\ref{const:3}) correspondingly. Denote with $L$ the whole set of
the Lagrange multipliers. Then, the dual function is formulated as
follows:
\begin{eqnarray}\label{eq:dual}
L_A(L)=\max_{\Lambda \in \mathcal{U}_{\lambda}^F, ~R \in
\mathcal{U}_{r}^{E\times (N-1)},~ P \in \mathcal{U}_{P}^{E}}
~\EE_{s,x}\left[\int_s^T \left(\sum_{i,d:i\in S_r(d)}
U_i^d(\lambda_i^d(t),t)-
\sum_{(i,j)}J_{ij}(P_{ij}(t))\right)dt\right]\nonumber \\
-\sum_{i,d} \mu_i^d \EE_{s,x}\left[\int_s^T \lambda_i^d(t)dt
+\int_s^T \sum_{j:i \in \mathcal{R}(j,d)} r_{ji}^d(t)dt-\int_s^T
\sum_{j\in \mathcal{R}(i,d)} r_{ij}^d(t)dt\right] \nonumber\\
-\sum_{(i,j)\in\mathcal{E}}l_{ij}\EE_{s,x}\left[\int_s^T \sum_{d}
r_{ij}^d(t)dt-\int_s^T C_{ij}(P(t))dt\right]- \sum_{i} \nu_i
\EE_{s,x}\left[\int_s^T\sum_{j\in \mathcal{N}^{out}_i}P_{ij}(t)dt-
P_{i,\max}\right].
\end{eqnarray}\noindent
Consequently, the dual problem of $P_1$ is defined as
$D_1:=\inf_{L}(L_A(L))$.
\begin{thm}
The problem $P_1$ has zero dual gap, i.e. if $P_1^*$ its optimal
value and $D_1^*$ the optimal value of the dual problem, then
$P_1^*=D_1^*$.\label{thm:1}
\end{thm}
To prove this theorem we proceed in analogy with Theorem $1$ in
\cite{ribeiro10}. The proof relies on the fact that the channel's
cumulative distribution function (cdf) is continuous and thus no
channel realization has strictly positive probability. It uses the
definition of nonatomic measures and the Lyapunov's convexity
theorem \cite{ribeiro10}.
\begin{IEEEproof}
To prove the zero duality gap, we consider a perturbed version of
the problem $P_1$, obtained by perturbing the constraints used to
define the Lagrangian. Let $P_1(\Delta)$ be the function that
assigns to each perturbation set
$\Delta=(\{\Delta_{i,d}^1\}_{\forall i,d},
\{\Delta_{i,j}^2\}_{\forall (i,j)}, \{\Delta_{i}^3\}_{\forall i})$,
the solution of the following perturbed optimization problem.
\begin{eqnarray}
P_1(\Delta)=\max_{\Lambda \in \mathcal{U}_{\lambda}^F, ~R \in
\mathcal{U}_{r}^{E\times (N-1)},~ P \in \mathcal{U}_{P}^{E}}
\EE_{s,x}\left[\int_s^T \left(\sum_{i,d:i\in S_r(d)}
U_i^d(\lambda_i^d(t),t)-
\sum_{(i,j)}J_{ij}(P_{ij}(t))\right)dt\right]\nonumber \\
s.t.\nonumber\\
\EE_{s,x}\left[\int_s^T \lambda_i^d(t)dt +\int_s^T \sum_{j:i \in
\mathcal{R}(j,d)} r_{ji}^d(t)dt-\int_s^T
\sum_{j\in \mathcal{R}(i,d)} r_{ij}^d(t)dt\right]\leq \Delta_{i,d}^1,~ \forall i,d \label{eq:constr1per}\\
\EE_{s,x}\left[\int_s^T \left(\sum_{d} r_{ij}^d(t)-
C_{ij}(P(t))\right)dt\right]\leq \Delta_{i,j}^2, ~\forall (i,j)\in
\mathcal{E}, ~ ~\EE_{s,x}\left[\int_s^T\sum_{j\in
\mathcal{N}^{out}_i}P_{ij}(t)dt\right]- P_{i,\max}\leq \Delta_{i}^3,
~\forall i \label{eq:constr3per}
\end{eqnarray}\noindent i.e. the constraints can be violated by
$\Delta$ amounts. In order to prove zero duality gap, we should show
that the function $P_1(\Delta)$ is a concave function of $\Delta$
\cite{ribeiro10}.

Let $\underline{\Delta}=(\{\underline{\Delta}_{i,d}^1\}_{\forall
i,d}, \{\underline{\Delta}_{i,j}^2\}_{\forall (i,j)},
\{\underline{\Delta}_{i}^3\}_{\forall i})$,
$\overline{\Delta}=(\{\overline{\Delta}_{i,d}^1\}_{\forall i,d},
\{\overline{\Delta}_{i,j}^2\}_{\forall (i,j)},
\{\overline{\Delta}_{i}^3\}_{\forall i})$ be two arbitrary sets of
perturbations with respective optimal values
$\underline{P_1}=P_1(\underline{\Delta})$,
$\overline{P_1}=P_1(\overline{\Delta})$ and respective solutions
$(\underline{\Lambda}, \underline{R}, \underline{P})$ and
$(\overline{\Lambda}, \overline{R}, \overline{P})$. Then, for an
arbitrary $a\in [0,1]$, we define the perturbation $\hat{\Delta}=a
\underline{\Delta}+(1-a) \overline{\Delta}$ and for feasible
solutions $(\hat{\Lambda}, \hat{R}, \hat{P})$, i.e. satisfying the
constraints (\ref{eq:constr1per}), (\ref{eq:constr3per}), we need to
show
\begin{equation} P_1(\hat{\Delta})= P_1(a \underline{\Delta}+(1-a)
\overline{\Delta})\geq a P_1(\underline{\Delta})+(1-a)
P_1(\overline{\Delta}).\label{eq:convex}\end{equation}

Consider the set of all possible state ($X$) realizations
$\mathcal{H}$ and the Borel field, $\mathbb{B}$, on $\mathcal{H}$.
For $A \in \mathbb{B}$, let $\EE_{s,x}^A$ be the expected value
restricted on channel realizations included in $A$. We define the
following measures.
\begin{eqnarray}
\theta_{id}(A)=\left[\EE_{s,x}^A\left[\int_s^T
U_i^d(\underline{\lambda}_i^d(t),t)dt\right],
\EE_{s,x}^A\left[\int_s^T U_i^d(\overline{\lambda}_i^d(t),t)dt
\right] \right],~ \forall i,d:i \in S_r(d),\\
\phi_{id}(A)=\left[\EE_{s,x}^A\left[\int_s^T
\underline{\lambda}_i^d(t)dt\right], \EE_{s,x}^A\left[\int_s^T
\overline{\lambda}_i^d(t)dt \right] \right],~ \forall i,d,\\
w_{ij}(A)=\left[\EE_{s,x}^A\left[\int_s^T
C_{ij}(\underline{P}(t))dt\right], \EE_{s,x}^A\left[\int_s^T
C_{ij}(\overline{P}(t))dt \right] \right],~ \forall (i,j)\in
\mathcal{E},\\ v_{ij}(A)=\left[\EE_{s,x}^A\left[\int_s^T
J_{ij}(\underline{P}_{ij}(t))dt\right], \EE_{s,x}^A\left[\int_s^T
J_{ij}(\overline{P}_{ij}(t))dt \right] \right],~ \forall (i,j) \in
\mathcal{E},\\ \xi_{ij}(A)=\left[\EE_{s,x}^A\left[\int_s^T
\underline{P}_{ij}(t)dt\right], \EE_{s,x}^A\left[\int_s^T
\overline{P}_{ij}(t)dt \right] \right],~ \forall (i,j) \in
\mathcal{E},
\end{eqnarray}\noindent while we also define
$\theta_{id}(\varnothing)=\phi_{id}(\varnothing)=w_{ij}(\varnothing)=v_{ij}(\varnothing)=\xi_{ij}(\varnothing)=0$.
These measures are nonatomic \cite{ribeiro10} since the channel cdf
is continuous and all the control variables are bounded. Thus, there
are no channel realizations with positive measure, i.e.,
$\theta_{id}(A)=\phi_{id}(A)=w_{ij}(A)=v_{ij}(A)=\xi_{ij}(A)=0$ for
all singleton sets $A \in \mathcal{H}$. Let $W(A)$ be the vector
measure expressing collectively all measures
$\theta_{id},\phi_{id},w_{ij},v_{ij},\xi_{ij}$, which is also
obviously nonatomic. Then from Lyapunov's convexity theorem
\cite{ribeiro10}, the range of $W$ is convex. Therefore, the value
$w_0=a W(\mathcal{H})+(1-a) W(\varnothing)=a W(\mathcal{H})$ belongs
to the range of values of $W$. As a result there exists $A_0 \in
\mathbb{B}$ such that $W(A_0)=a W(\mathcal{H})$, i.e.,
$\theta_{id}(A_0)=a \theta_{id}(\mathcal{H}),~\phi_{id}(A_0)=a
\phi_{id}(\mathcal{H})$, etc. Moreover, due to the additivity of
measures, for the complement of $A_0$, $A_0^c$, we have $W(A_0^c)=
W(\mathcal{H})-W(A_0)=(1-a)W(\mathcal{H})$. Then, we define the
following controls for the new perturbation $\hat{\Delta}$.
\begin{eqnarray}
\hat{r}_{ij}^d(t)=a \underline{r}_{ij}^d(t)+ (1-a)
\overline{r}_{ij}^d(t), ~\forall d,(i,j)\in \mathcal{E}, t \in
[s,T], ~\mathbf{P}-a.s.
\end{eqnarray}
\begin{minipage}{0.17\linewidth}
\begin{eqnarray}\hat{\lambda}_{i}^d(t)=\begin{cases}\underline{\lambda}_{i}^d(t)
~\text{within}~ A_0 \nonumber\\
\overline{\lambda}_{i}^d(t)~ \text{within}~ A_0^c
\end{cases}
\end{eqnarray}\end{minipage}\begin{minipage}{1\linewidth}
\begin{eqnarray}\hat{P}_{ij}(t)=\begin{cases}\underline{P}_{ij}(t)
~\text{within}~ A_0\nonumber \\ \overline{P}_{ij}(t)~ \text{within}~
A_0^c
\end{cases}
\forall i,d:i \in S_r(d), \forall (i,j)\in \mathcal{E}, t \in
[s,T],~ \mathbf{P}-a.s.
\end{eqnarray}\end{minipage}\noindent Obviously since $\underline{\Lambda},~ \overline{\Lambda}~ \in \mathcal{U}_{\lambda}^F, ~\underline{R}, ~\overline{R} ~\in
\mathcal{U}_{r}^{E\times (N-1)},~ \underline{P}, ~\overline{P}~ \in
\mathcal{U}_{P}^{E}$, it also holds that $\hat{\Lambda} \in
\mathcal{U}_{\lambda}^F, ~\hat{R} \in \mathcal{U}_{r}^{E\times
(N-1)},~ \hat{P} \in \mathcal{U}_{P}^{E}$. Now, based on the above,
we check if the controls defined above for the perturbation
$\hat{\Delta}$ satisfy Ineqs. (\ref{eq:constr1per}),
(\ref{eq:constr3per}) and if Ineq. (\ref{eq:convex}) holds. For the
constraint (\ref{eq:constr1per}), we have:
\begin{eqnarray}
\EE_{s,x}\left[\int_s^T \left(\hat{\lambda}_i^d(t)+\sum_{j:i \in
\mathcal{R}(j,d)} \hat{r}_{ji}^d(t)- \sum_{j\in \mathcal{R}(i,d)}
\hat{r}_{ij}^d(t)\right)dt\right]= \EE_{s,x}^{A_0}\left[\int_s^T
\underline{\lambda}_i^d(t)dt\right] +
\EE_{s,x}^{A_0^c}\left[\int_s^T
\overline{\lambda}_i^d(t)dt\right]+\nonumber\\
+\EE_{s,x}\left[\int_s^T \sum_{j:i \in \mathcal{R}(j,d)}\left( a
\underline{r}_{ji}^d(t)+ (1-a)
\overline{r}_{ji}^d(t)\right)dt-\int_s^T \sum_{j\in
\mathcal{R}(i,d)} \left(a \underline{r}_{ij}^d(t)+ (1-a)
\overline{r}_{ij}^d(t)\right)dt\right]=\nonumber\\
a\EE_{s,x}\left[\int_s^T \underline{\lambda}_i^d(t)dt +\int_s^T
\sum_{j:i \in \mathcal{R}(j,d)} \underline{r}_{ji}^d(t)dt -\int_s^T
\sum_{j\in \mathcal{R}(i,d)}
\underline{r}_{ij}^d(t)dt\right]\nonumber\\+(1-a)\EE_{s,x}\left[\int_s^T
\overline{\lambda}_i^d(t)dt +\int_s^T \sum_{j:i \in
\mathcal{R}(j,d)} \overline{r}_{ji}^d(t)dt -\int_s^T \sum_{j\in
\mathcal{R}(i,d)} \overline{r}_{ij}^d(t)dt\right]\leq
a\underline{\Delta}_{i,d}^1 +(1-a) \overline{\Delta}_{i,d}^1=
\hat{\Delta}_{i,d}^1.
\end{eqnarray}
With respect to the first constraint in (\ref{eq:constr3per}), we
have:
\begin{eqnarray}
\EE_{s,x}\left[\int_s^T \sum_{d} \hat{r}_{ij}^d(t)dt\right]-
\EE_{s,x}\left[\int_s^T C_{ij}(\hat{P}(t))dt\right]=%
a\EE_{s,x}\left[\int_s^T \sum_{d} \underline{r}_{ij}^d(t)dt-\int_s^T
C_{ij}(\underline{P}(t))dt\right]\nonumber \\+
(1-a)\EE_{s,x}\left[\int_s^T \sum_{d}
\overline{r}_{ij}^d(t)dt-\int_s^T
C_{ij}(\overline{P}(t))dt\right]\leq a\underline{\Delta}_{i,j}^2
+(1-a) \overline{\Delta}_{i,j}^2= \hat{\Delta}_{i,j}^2.
\end{eqnarray}
Similarly, with respect to the second constraint in
(\ref{eq:constr3per}), we have:
\begin{eqnarray}
\EE_{s,x}\left[\int_s^T\sum_{j\in
\mathcal{N}^{out}_i}\hat{P}_{ij}(t)dt\right]- P_{i,\max}=
\EE_{s,x}^{A_0}\left[\int_s^T\sum_{j\in \mathcal{N}^{out}_i}
\underline{P}_{ij}(t)dt\right]+\EE_{s,x}^{A_0^c}\left[\int_s^T\sum_{j\in
\mathcal{N}^{out}_i} \overline{P}_{ij}(t)dt\right]-
P_{i,\max}=\nonumber\\ a\EE_{s,x}\left[\int_s^T\sum_{j\in
\mathcal{N}^{out}_i}\underline{P}_{ij}(t)dt\right]+(1-a)\EE_{s,x}\left[\int_s^T\sum_{j\in
\mathcal{N}^{out}_i}\overline{P}_{ij}(t)dt\right]- P_{i,\max}\leq
a\underline{\Delta}_{i}^3 +(1-a) \overline{\Delta}_{i}^3=
\hat{\Delta}_{i}^3.
\end{eqnarray}
Finally, $P_1(\hat{\Delta})\geq \EE_{s,x}\left[\int_s^T
\left(\sum_{i,d:i\in S_r(d)} U_i^d(\hat{\lambda}_i^d(t),t)-
\sum_{(i,j)}J_{ij}(\hat{P}_{ij}(t))\right)dt\right]$, i.e.,
\begin{eqnarray}P_1(\hat{\Delta})\geq= \EE_{s,x}^{A_0}\left[\int_s^T \left(\sum_{i,d:i\in
S_r(d)}
U_i^d(\underline{\lambda}_i^d(t),t)-\sum_{(i,j)}J_{ij}(\underline{P}_{ij}(t))\right)dt\right]\nonumber\\+
\EE_{s,x}^{A_0^c}\left[\int_s^T \left( \sum_{i,d:i\in S_r(d)}
U_i^d(\overline{\lambda}_i^d(t),t)-\sum_{(i,j)}J_{ij}(\overline{P}_{ij}(t))\right)dt\right]=
aP_1(\underline{\Delta})+(1-a)P_1(\overline{\Delta}),
\end{eqnarray}\noindent which concludes the proof.
\end{IEEEproof}

\subsection{Problem Solution}
Since the dual gap corresponding to the problem $P_1$ is null, we
can obtain its optimal value by solving its dual problem via a
subgradient methodology \cite{bertsekas09}. The subgradient
algorithm, repeatedly renews the Lagrange multipliers until
converging to their optimal solutions. We use the symbol
$\eta=\{0,1,..\}$ for the repetitions of the subgradient algorithm.
Then, the renewal equations of the Lagrange multipliers take the
form:
\begin{eqnarray}\label{eq:lagrangea}
\mu_i^d(\eta+1)=\left\{\mu_i^d(\eta)+\kappa(\eta)\cdot
\EE_{s,x}\left[\int_s^T \lambda_i^{d*}(t)dt +\int_s^T \sum_{j:i \in
\mathcal{R}(j,d)} r_{ji}^{d*}(t)dt-\int_s^T
\sum_{j\in \mathcal{R}(i,d)} r_{ij}^{d*}(t)dt\right]\right\}^+,~ \forall i,d,\\
l_{ij}(\eta+1)=\left\{l_{ij}(\eta)+\kappa(\eta)\cdot
\EE_{s,x}\left[\int_s^T \sum_{d} r_{ij}^{d*}(t)dt-\int_s^T
C_{ij}(P^*(t))dt\right]\right\}^+, ~\forall (i,j)\in\mathcal{E},\label{eq:lagrangeb} \\
\nu_i(\eta+1)=\left\{\nu_i(\eta)+\kappa(\eta)\cdot
\EE_{s,x}\left[\int_s^T\sum_{j\in \mathcal{N}^{out}_i}P_{ij}^*(t)dt-
P_{i,\max}\right]\right\}^+,~\forall i.\label{eq:lagrangec}
\end{eqnarray}\noindent where $\left\{\right\}^+$ denotes projection to $[0,\infty)$ and the values $\mu_i^d(0)\geq 0, ~\forall i,d, ~l_{ij}(0)\geq 0,~\forall (i,j)\in
\mathcal{E},~\nu_i(0)\geq 0, ~\forall i$, are considered given. The
subgradient method is known to converge to a close neighborhood of
the optimal values for the Lagrange multipliers if constant
step-size, $\kappa(\eta)$, is used, while diminishing, non summable
but square summable step size allows for convergence to the optimal
values \cite{bertsekas09}, \cite{gatsis11}. The values of the
stochastic processes $\lambda_i^{d*},~\forall
i,d,~r_{ij}^{d*},~\forall (i,j),d,~P_{ij}^*,~\forall (i,j),$ are
computed while obtaining the dual function (Eq. (\ref{eq:dual}))
with the current solution for the Lagrange multipliers, i.e. for the
iteration $\eta$. For performing this maximization, we can observe
that one cannot achieve a better objective value than the one
achieved by choosing at each time $t\in[s,T]$, each control variable
optimally as a function of $X(t)$ and the current values of the
Lagrange multipliers (iteration $\eta$). Therefore, at the iteration
$\eta$, given $\mu_i^d(\eta), ~\forall i,d, ~l_{ij}(\eta)~\forall
(i,j)\in \mathcal{E},~\nu_i(\eta), ~\forall_i$, the controls are
computed as follows:
\begin{itemize}
\item The optimal $\lambda_i^{d*},~\forall i,d: i \in S_r(d)$, at the transport layer, are computed source-wise by
\begin{equation}\label{eq:lamda}
\frac{\vartheta U_i^d(\lambda_i^d,t)}{\vartheta
\lambda_i^d}-\mu_i^d(\eta)=0, \forall~t\in [s,T],
\end{equation}\noindent while taking into account that $\lambda_i^{d*} \in U_{\lambda}$. From Eq. (\ref{eq:lamda}), it is
observed that $\lambda_i^{d*}$ is time-varying but not random as it
depends only on the deterministic Lagrange multiplier
$\mu_i^d(\eta)$ (which is constant in time $t\in [s,T]$) and the
deterministic time-varying utility function $U_i^d$ but not on the
channel state $X(t)$. In the case that
$U_i^d(\lambda_i^d,t),~\forall i,d:i\in S_r(d)$ are invariant with
$t$ then $\lambda_i^{d*},~\forall i,d:i\in S_r(d)$ are constants
over $[s,T]$, too.
\item The optimal routing variables $r_{ij}^{d*},~\forall (i,j),d$ (network layer) are computed by
\begin{equation}\label{eq:r}
\max_{R}\sum_d \sum_{(i,j)|j\in
R(i,d)}r_{ij}^d(\mu_i^d(\eta)-\mu_j^d(\eta)-l_{ij}(\eta)),~
\forall~t\in [s,T],
\end{equation}\noindent i.e. $r_{ij}^{d*}(t)=R_{\max},~\forall~t\in [s,T]$, if $(\mu_i^d(\eta)-\mu_j^d(\eta)-l_{ij}(\eta))>0$, while it is
observed (from Eq. (\ref{eq:r})) that each $r_{ij}^{d*}$ is constant
in time.
\item The optimal transmission power values at the physical layer, $P^*$, are
computed via solving:
\begin{equation}\label{eq:p}
\max_{P}
\sum_{(i,j)}\left(-J_{ij}(P_{ij})+l_{ij}(\eta)C_{ij}(P)-\nu_i(\eta)
P_{ij}\right), \forall~t\in [s,T],
\end{equation}\noindent while considering the $U_{P}$ constraints (Section \ref{sec:sysmodel}). From Eq. (\ref{eq:p}), it is
observed that $P$ is a stochastic process since the link capacity
$C_{ij},~\forall (i,j)$ depends on the stochastic process $X(t)$
representing the wireless channels (Section \ref{sec:sysmodel}).
\end{itemize}
The computation of the expected values involved in the Lagrange
multipliers' renewal equations requires Monte Carlo simulations.
However, since the controls $\Lambda^*,~R^*$, have been shown to be
deterministic, the expected value involved in the renewal equation
of the Lagrange multipliers $\mu_i^d, ~\forall i,d$ (Eq.
(\ref{eq:lagrangea})) is superfluous. For the rest of the Lagrange
multipliers, the expected values are obtained via the following
algebraic computations.

Firstly, the processes \(\{X_{ij}(t):\ t\ge 0\}\) are discretized \cite{higham01}. We
compute $\delta t=\frac{T-s}{n}$, where $n$ is a design parameter
representing the number of samples of the channel over the time
interval $[s,T]$ of the network's operation, and we sample on 
the time instants $\{\tau_b=s+b \cdot \delta t\}_{b=1:n}$. In view of
(\ref{eq:sdepathloss}) it is not hard to see that for every $i,j$ and $b$ we have
\[
X_{ij}(\tau_b)=\rho_{ij}(b)X_{ij}(\tau_{b-1})+\zeta_{ij}(b)+Z_{ij}(b)
\]
where 
\[
\rho_{ij}(b)=\exp(-\int_{\tau_{b-1}}^{\tau_b}\beta_{ij}(s)ds)
\]
\[
\zeta_{ij}(b)=\int_{\tau_{b-1}}^{\tau_b}\beta_{ij}(s)\gamma_{ij}(s)\exp\big(-\int_{s}^{\tau_b}\beta_{ij}(q)dq\big)ds
\]
and
\[
Z_{ij}(b)=\int_{\tau_{b-1}}^{\tau_b}\delta_{ij}(s)\exp\big(-\int_{s}^{\tau_b}\beta_{ij}(q)dq\big)dW_{ij}(s).
\]
Note that $\{Z_{ij}(b)\}_{i,j,b}$ are independent random variables with distributions ${\cal N}(0,\sigma_{ij}^2(b))$, where
\[
\sigma_{ij}^2(b)=\int_{\tau_{b-1}}^{\tau_b}\delta_{ij}^2(s)\exp(-2\int_{s}^{\tau_b}\beta_{ij}(q)dq)ds.
\]
The discretized scheme then becomes for every $(i,j)$:
\begin{equation}
\label{eq:sdenumerical}
X_{ij}(\tau_b)=\rho_{ij}(b)X_{ij}(\tau_{b-1})+\zeta_{ij}(b)+\sigma_{ij}(b) \xi_{ij}(b),\ b=1,2,\ldots,n,\quad X_{ij}(s)=x_{ij0}
\end{equation}
where $\{\xi_{ij}(b)\}_{i,j,b}$ are independent samples from a standard normal random variable. After numerically
computing the solution of the SDE (\ref{eq:sdepathloss}), we compute
$P^*(\tau_b)$ for each $b$ from Eq. (\ref{eq:p}) and we use a
Riemann sum approximation for a sample of the integral $\int_{s}^{T}
C_{ij}(P^*(t))dt$, that is
 $\int_{s}^{T} C_{ij}(P^*(t))dt\simeq\sum_{b=0}^{n-1}
C_{ij}(P^*(\tau_b))\delta t$, where
${C_{ij}}({P^*}({\tau _b})) = {B_{ij}}{\log _2}\left( {1 +
\frac{{{e^{K{X_{ij}}({\tau _b})}}P_{ij}^*({\tau _b})}}{{{N_0} +
\sum\limits_{(k,l) \in {{\cal I}_{ij}}} {{e^{K{X_{kj}}({\tau
_b})}}P_{kl}^*({\tau _b})} }}} \right)$ (Sections
\ref{sec:channelmodels}, \ref{sec:sysmodel}). Finally, we repeat the
above procedure $M$ times to obtain $M$ independent samples of the
preceding integral, and we average these samples to estimate the
expected capacity of link $(i,j)$ over the time interval $[s,T]$,
i.e. $\EE_{s,x}\left[\int_{s}^{T} C_{ij}(P^*(t))dt\right]$,
appearing in the Eq. (\ref{eq:lagrangeb}). 
Similarly we obtain
$\EE_{s,x}\left[\int_{s}^{T} P_{ij}^*(t)dt\right]$. Since the
computation of $P^*$ from Eq. (\ref{eq:p}) involves the Lagrange
multipliers' values, the Monte Carlo computations of the expected
values should be performed at each iteration of the subgradient
algorithm. Notably, $\rho_{ij}(b),\zeta_{ij}(b),\sigma_{ij}(b)$ are deterministic so we only need to compute them once. 

The optimal power allocation problem at the physical layer, i.e.,
the solution of Eq. (\ref{eq:p}) determines the complexity of the
whole problem since everything else is simple algebraic
computations. Indeed the computations of Eqs. (\ref{eq:lamda}),
(\ref{eq:r}) can be distributed to the sources and links
correspondingly. This cross-terminal optimization problem at the
physical layer constitutes an important challenge in wireless
networking \cite{ribeiro10} which is treated in other works in
literature \cite{ribeiro10}, \cite{gatsis11} and is out of the scope
of this paper. In the next section, we formulate and solve the same
problem in the case of orthogonal access to the wireless medium
where the capacity functions take much simpler concave forms leading
to tractable analytic solutions.

We summarize below the steps of the algorithm proposed in this
section for obtaining $D_1^*$.
\begin{enumerate}
\item Initialize the Lagrange multipliers, $\eta=0,~\mu_i^d(0), ~\forall i,d, ~l_{ij}(0),~\forall (i,j)\in
\mathcal{E},~\nu_i(0), ~\forall_i$.
\item Compute $\lambda_i^{d*},~\forall
i,d,~r_{ij}^{d*},~\forall (i,j),d$ using Eqs. (\ref{eq:lamda}),
(\ref{eq:r}) respectively.
\item Compute the expected values $\EE_{s,x}\left[\int_{s}^{T} C_{ij}(P^*(t))dt\right],~\EE_{s,x}\left[\int_{s}^{T} P_{ij}^*(t)dt\right],~\forall (i,j)$, as described.
\item Compute $\mu_i^d(\eta+1), ~\forall i,d, ~l_{ij}(\eta+1),~\forall (i,j)\in
\mathcal{E},~\nu_i(\eta+1), ~\forall_i$ from Eqs.
(\ref{eq:lagrangea}), (\ref{eq:lagrangeb}), (\ref{eq:lagrangec}) and
set $\eta \leftarrow \eta+1$.
\item Repeat steps $2,~3,~4$ until convergence.
\end{enumerate}
\subsection{Discussion} It is important to note that the time scale of the
renewal of the Lagrange multipliers should be distinguished from the
time interval $[s,T]$ of the network's operation. In principle, the
above algorithm should run off-line, i.e. prior to the network
operation to determine the optimal source rates, routing variables
and Lagrange multipliers and afterwards, the online network
operation will be designed based on these optimal values and the
solution of the cross-terminal power allocation problem of Eq.
(\ref{eq:p}). Note that convergence of the dual variables close to
their optimal values does not imply convergence of the primal
variables except if the primal variables change continuously with
respect to the optimal Lagrange multipliers (e.g. source rates).
Following the approach of \cite{gatsis11}, we can compute optimal
routing variables while performing the subgradient iterations.
Specifically let as assume that $N_o$ is the total number of
subgradient iterations, while the index $\eta \in {0,...,N_o-1}$ is
used to distinguish each one iteration. Then, if applying
$\bar{r}_{ij}^d(t,N_o)=\frac{\sum_{\eta=0}^{N_o-1}r_{ij}^{d*}(t,\eta)}{N_o},~\forall
(i,j),d$ as optimal routing variables for each $t\in [s,T]$, we can
achieve a close to the optimal value of $P_1$, using diminishing
step size. This can be proven as in \cite{gatsis11}, if we first
reformulate $P_1$ in an equivalent form replacing the objective
function by the optimization variable $P'$ and adding the constraint
$\EE_{s,x}\left[\int_s^T \left(\sum_{i,d:i\in S_r(d)}
U_i^d(\lambda_i^d(t),t)-
\sum_{(i,j)}J_{ij}(P_{ij}(t))\right)dt\right]\geq P'$. Then,
obviously, Theorem \ref{thm:1} holds. More discussion on a possible
(suboptimal) online implementation of the proposed algorithm is made
in Section \ref{sec:simulation}.

We also note that instantaneous values of the controls (online
approach - as functions of $X(t)$) during the network operation for
such a problem may be obtained via a dynamic programming solution
methodology (Hamilton-Jacobi-Bellman partial differential equation
(HJB pde)) \cite{fleming06} which adds dramatically to complexity
for a wireless multihop network (specifically the solution of the
HJB pde is completely inefficient \cite{cui14}, \cite{olama06}).

\section{Orthogonal Access to the Medium}\label{sec:orthogonal}
In this section, we redesign the problem $P_1$ allowing only
orthogonal access to the wireless medium, and thus leading to convex
link capacity forms (Section \ref{sec:sysmodel}) since the noise
from interference will become negligible. In order to achieve this,
we introduce new optimization variables for each independent set
$\iota$, denoted as $\pi_{\iota}$, expressing the activation
percentage of the corresponding independent set at time $t\in
[s,T]$, and further satisfying the relations: $ \sum_{\iota}
\pi_{\iota}(t) \leq 1, ~0\leq \pi_{\iota}(t) \leq
1,~\forall \iota, t \in [s,T].$ 
 In the following, $\Pi$ stands for the collection of all
$\pi_{\iota},~\forall \iota$ and $I_n$ is the number of the
independent sets of the network's connectivity graph. Since the
channel state is a stochastic process, similarly to the definition
of the rest of the control variables (Section \ref{sec:sysmodel}),
we define the value range for each $\pi_{\iota}$, $U_{\pi}=[0,1]$,
and the corresponding feasible set
$\mathcal{U}_{\pi}=\left\{\pi:[s,T]\times \Omega\rightarrow U_{\pi}:
\pi \mathrm{~is~} \{\mathcal{F}_t\}_{t\geq s} \mathrm{~adapted}
\right\}$. Then, $\Pi \in
\mathcal{U}_{\Pi}=\{\mathcal{U}_{\pi}^{I_n}: \sum_{\iota=1}^{I_n}
\pi_{\iota}(t) \leq 1, ~\forall t \in [s,T]\}$.

The new optimization variables impose time-sharing among the
independent sets, thus they render the interference levels
negligible and the capacity of each link $(i,j)$ is given by the
concave function in Section \ref{sec:sysmodel}. The time share
corresponding to the link $(i,j)$ at $t\in [s,T]$, is given by
$\sum_{\iota: (i,j) \in \iota} \pi_{\iota}(t)$. The new optimization
problem $\mathbf{P_2}$ is formulated as:
\begin{eqnarray}
P_2:=\max_{\Lambda \in \mathcal{U}_{\lambda}^F, ~R \in
\mathcal{U}_{r}^{E\times (N-1)},~ P \in \mathcal{U}_{P}^{E}, ~\Pi
\in \mathcal{U}_{\Pi}} \EE_{s,x}\left[\int_s^T \left(\sum_{i,d:i\in
S_r(d)} U_i^d(\lambda_i^d(t),t)- \sum_{(i,j)}\sum_{\iota: (i,j) \in
\iota} \pi_{\iota}(t)J_{ij}(P_{ij}(t))\right)dt\right]\nonumber\\
s.t.\nonumber\end{eqnarray}
\begin{eqnarray}
\EE_{s,x}\left[\int_s^T \lambda_i^d(t)dt +\int_s^T \sum_{j:i \in
\mathcal{R}(j,d)} r_{ji}^d(t)dt\right]\leq \EE_{s,x}\left[\int_s^T
\sum_{j\in \mathcal{R}(i,d)} r_{ij}^d(t)dt\right],~ \forall i,d \label{const:1b}\\
\EE_{s,x}\left[\int_s^T \sum_{d} r_{ij}^d(t)dt\right]\leq
\EE_{s,x}\left[\int_s^T \sum_{\iota: (i,j)
\in \iota} \pi_{\iota}(t)C_{ij}(P_{ij}(t))dt\right], ~\forall (i,j)\in \mathcal{E}\label{const:2b}\\
\EE_{s,x}\left[\int_s^T\sum_{j\in \mathcal{N}^{out}_i}\sum_{\iota:
(i,j) \in \iota} \pi_{\iota}(t)P_{ij}(t)dt\right]\leq P_{i,\max},
~\forall i\label{const:3b}
\end{eqnarray}

The formulation  of $P_2$ is similar to the one of $P_1$ (Section
\ref{sec:nonortho}), with the difference that in $P_2$, we have
introduced the new optimization variables $\Pi \in
\mathcal{U}_{\Pi}$, the link capacities are concave and the link
transmission powers, $P_{ij}(t)$, the link costs with respect to the
transmission powers, $J_{ij}(P_{ij}(t))$, and the link capacities,
$C_{ij}(P_{ij}(t))$, are replaced by their effective values
$\sum_{\iota: (i,j) \in \iota} \pi_{\iota}(t)P_{ij}(t)$,
$\sum_{\iota: (i,j) \in \iota} \pi_{\iota}(t)J_{ij}(P_{ij}(t))$,
$\sum_{\iota: (i,j) \in \iota} \pi_{\iota}(t)C_{ij}(P_{ij}(t))$,
correspondingly \cite{marques}. $P_2$ is non-convex due to the
appearance of the control variables in multiplicative form in the
objective function and the constraints (Eqs. (\ref{const:2b}),
(\ref{const:3b})) in addition to the general forms of the utility
and cost functions. However, in a similar way as for $P_1$, it can
be shown that $P_2$ has a zero duality gap. Let $\mu_i^d\geq 0,
~\forall i,d, ~l_{ij}\geq 0,~\forall (i,j)\in \mathcal{E},~\nu_i\geq
0, ~\forall i$, be the Lagrange multipliers associated with the
constraints (\ref{const:1b}), (\ref{const:2b}), (\ref{const:3b}),
respectively. Denote with $L$ the whole set of the Lagrange
multipliers. Then, the dual function is formulated as follows:
\begin{eqnarray}\label{eq:dual2}
L_A(L)=\max_{\Lambda \in \mathcal{U}_{\lambda}^F, ~R \in
\mathcal{U}_{r}^{E\times (N-1)},~ P \in \mathcal{U}_{P}^{E},~\Pi \in
\mathcal{U}_{\Pi}} ~\EE_{s,x}\left[\int_s^T \left(\sum_{i,d:i\in
S_r(d)} U_i^d(\lambda_i^d(t),t)- \sum_{(i,j)}\sum_{\iota: (i,j) \in
\iota}
\pi_{\iota}(t)J_{ij}(P_{ij}(t))\right)dt\right]\nonumber \\
-\sum_{i,d} \mu_i^d \EE_{s,x}\left[\int_s^T \lambda_i^d(t)dt
+\int_s^T \sum_{j:i \in \mathcal{R}(j,d)} r_{ji}^d(t)dt-\int_s^T
\sum_{j\in \mathcal{R}(i,d)} r_{ij}^d(t)dt\right]
-\nonumber\\\sum_{(i,j)\in\mathcal{E}}l_{ij}\EE_{s,x}\left[\int_s^T
\left(\sum_{d} r_{ij}^d(t)-\sum_{\iota: (i,j) \in \iota}
\pi_{\iota}(t) C_{ij}(P_{ij}(t))\right)dt\right]- \sum_{i} \nu_i
\EE_{s,x}\left[\int_s^T\sum_{j\in \mathcal{N}^{out}_i}\sum_{\iota:
(i,j) \in \iota} \pi_{\iota}(t)P_{ij}(t)dt- P_{i,\max}\right].
\end{eqnarray}\noindent
Then, the dual problem is defined as $D_2:=\inf_{L}(L_A(L))$.
\begin{thm}
The problem $P_2$ has zero dual gap.\label{thm:2}
\end{thm}
The proof is briefly described in Appendix B, provided as a
supplementary file, as it is very similar to the proof of Theorem
\ref{thm:1}. Now, we obtain the optimal value of the problem $P_2$
via solving its dual. The renewal equations of the Lagrange
multipliers are given by ($\eta=\{0,1,2..\}$):
\begin{eqnarray}\label{eq:lagrange2}
\mu_i^d(\eta+1)=\left\{\mu_i^d(\eta)+\kappa(\eta)\cdot
\EE_{s,x}\left[\int_s^T \lambda_i^{d*}(t)dt +\int_s^T \sum_{j:i \in
\mathcal{R}(j,d)} r_{ji}^{d*}(t)dt-\int_s^T \sum_{j\in
\mathcal{R}(i,d)} r_{ij}^{d*}(t)dt\right]\right\}^+,~ \forall
i,d,\end{eqnarray}
\begin{eqnarray}
l_{ij}(\eta+1)=\left\{l_{ij}(\eta)+\kappa(\eta)\cdot
\EE_{s,x}\left[\int_s^T \sum_{d} r_{ij}^{d*}(t)dt-\int_s^T
\sum_{\iota: (i,j) \in \iota} \pi_{\iota}^*(t)C_{ij}(P_{ij}^*(t))dt\right]\right\}^+, ~\forall (i,j)\in\mathcal{E}, \label{eq:lagrange2b}\\
\nu_i(\eta+1)=\left\{\nu_i(\eta)+\kappa(\eta)\cdot
\EE_{s,x}\left[\int_s^T\sum_{j\in \mathcal{N}^{out}_i}\sum_{\iota:
(i,j) \in \iota} \pi_{\iota}^*(t)P_{ij}^*(t)dt-
P_{i,\max}\right]\right\}^+,~\forall i.\label{eq:lagrange2c}
\end{eqnarray}\noindent with given values $\mu_i^d(0)\geq 0, ~\forall i,d, ~l_{ij}(0)\geq 0,~\forall (i,j)\in
\mathcal{E},~\nu_i(0)\geq 0, ~\forall_i$.

The optimal values of the stochastic processes
$\lambda_i^{d*},~\forall i,d,~r_{ij}^{d*},~\forall (i,j),d$, for
each $\eta$ are computed by Eqs. (\ref{eq:lamda}), (\ref{eq:r})
correspondingly and the same observations hold. Regarding the
optimal values $P_{ij}^*,~\forall (i,j)\in
\mathcal{E},~\pi_{\iota}^*,~\forall \iota$, for each $\eta$, they
are obtained by solving $\forall~t\in [s,T]$:
\begin{eqnarray}\label{eq:p2}
\max_{P,\Pi} \sum_{(i,j)}\sum_{\iota: (i,j) \in \iota} \pi_{\iota}
\left(l_{ij}(\eta)C_{ij}(P_{ij})-J_{ij}(P_{ij})-\nu_i(\eta)
P_{ij}\right)= \max_{P,\Pi} \sum_{\iota}\pi_{\iota} \sum_{(i,j) \in
\iota} \left(l_{ij}(\eta)C_{ij}(P_{ij})-J_{ij}(P_{ij})-\nu_i(\eta)
P_{ij}\right),
\end{eqnarray}\noindent which constitutes a maximum weight matching problem
over the independent sets. Specifically, for its solution, each link
$(i,j)$ computes the stochastic process $P_{ij}^*(t)$, $t\in[s,T]$
(depending on the state's, $X$, path) by maximizing
$-J_{ij}(P_{ij})+l_{ij}(\eta)C_{ij}(P_{ij})-\nu_i(\eta) P_{ij}$,
i.e. solving
\begin{equation}\label{eq:p3}
-\frac{\vartheta J_{ij}(P_{ij})}{\vartheta
P_{ij}}+l_{ij}(\eta)\frac{\vartheta C_{ij}(P_{ij})}{\vartheta
P_{ij}}-\nu_i(\eta) =0,~ \forall~t\in [s,T],
\end{equation}\noindent while taking into account the $U_P$ constraints (Section \ref{sec:sysmodel}). Then, each link $(i,j)\in \mathcal{E}$ is assigned a weight equal
to
$W_e(i,j,t)=\left(-J_{ij}(P_{ij}^*(t))+l_{ij}(\eta)C_{ij}(P_{ij}^*(t))-\nu_i(\eta)
P_{ij}^*(t)\right)$ and finally the independent set $\iota^*$
maximizing the sum $\sum_{(i,j) \in \iota}W_e(i,j,t)$ receives
$\pi_{\iota^*}(t)=1$, while $\pi_{\iota}(t)=0$ for $\iota\neq
\iota^*$ at $t\in [s,T]$. In this paper, we assume that ties break
arbitrarily, however, a study on how to break ties can be found in
\cite{marques}.

The computation of the expected values involved in the Lagrange
multipliers follows the lines of the Monte Carlo simulations
described in Section \ref{sec:nonortho} and the algorithm for
solving $D_2$ is similar to the one in Section \ref{sec:nonortho}.
Note that Eq. (\ref{eq:p3}) can be solved link-wise in a very
efficient manner, contrary to Eq. (\ref{eq:p}) which involves the
complex cross-terminal problem. The observations regarding the
convergence to the optimal Lagrange multipliers and primal values
are of similar nature to the ones of Section \ref{sec:nonortho}. At
this point we study the solution of Eq. (\ref{eq:p3}) in more detail
in order to gain more insight regarding the optimal power control.
Let us assume LTF and convex link costs of the form
$J_{ij}(P_{ij})=V P_{ij}^2$, where $V>0$ is a constant for all
$(i,j)\in \mathcal{E}$. Then, for a given $\eta$, the solution of
Eq. (\ref{eq:p3}) is explicitly given by:
\begin{eqnarray}\label{eq:p3solution}
P_{ij}^*(X_{ij}(t))=
\max\big\{0,\min\{\tilde{P}_{ij}^*(X_{ij}(t)),P_{\max}\}\big\},~\mathrm{where}
\nonumber \end{eqnarray}
\begin{eqnarray}
\tilde{P}_{ij}^*(X_{ij}(t))=\frac{1}{2}\left[-\left(N_0
e^{-KX_{ij}(t)}+\frac{\nu_i(\eta)}{2V} \right)+\sqrt{\left(N_0
e^{-KX_{ij}(t)}+\frac{\nu_i(\eta)}{2V} \right)^2-4\left(
\frac{\nu_i(\eta) N_0e^{-KX_{ij}(t)}}{2V}-\frac{l_{ij}(\eta)
B_{ij}}{2V \log(2)}\right)}\right]
\end{eqnarray}
Note that the optimally controlled power at $\eta$ never exceeds the
value
 \begin{small}
\begin{eqnarray}
\label{eq:optPowermax}P_{ij}^{* \max}=\frac{1}{2}\left(\sqrt{\left(
\frac{\nu_i(\eta)}{2V}\right)^2+\frac{4l_{ij}(\eta) B_{ij}}{2V
\log(2)}} -\frac{\nu_i(\eta)}{2V} \right),
\end{eqnarray}
\end{small}\noindent irrespectively of the $P_{\max}$ value, while this value is achieved asymptotically
when $X_{ij}(t)\rightarrow -\infty$. Note also that
$P_{ij}^*(X_{ij}(t))\rightarrow 0$ when $X_{ij}(t)\rightarrow
\infty$. Therefore, the optimal power control exploits low power
loss values created by random fluctuations to increase the link
capacity as much as possible, thus positively affecting the flows'
throughput (rates). On the other hand, transmission power is not
wasted when the power loss is high. Note that in case of orthogonal
access to the medium the aim of power control is not scheduling
(which is defined via the $\Pi$ variables) but to take advantage of
the channel when it is favorable and to avoid depleting energy when
the channel is destructive.


In the following, we assume no power control and scheduling and we
prove an interesting counter-intuitive theorem regarding the
relation between the achieved utility by the network and the
channel's diffusion coefficient in the case of LTF. Absence of power
control means that in $P_2$, $J_{ij}(t)=0,~\forall (i,j), t \in
[s,T]$, the constraint of Eq. (\ref{const:3b}) is dropped and
finally, $P_{ij}(t),~\forall (i,j), t \in [s,T]$ are constant and
predefined. Absence of scheduling means that the time share of each
link is constant and predefined, e.g. $\sum_{\iota: (i,j) \in \iota}
\pi_{\iota}(t)=\zeta_{ij}\geq 0,~\forall (i,j), t \in [s,T]$.

\subsection{Utilities as Increasing Functions of the Channel's Diffusion Coefficient in the case of LTF}
\label{sec:utilities} We examine the effect of the channel's
diffusion coefficient, on the users' optimal utility. Contrary to
what is perhaps expected, we prove that an increase of a channel's
diffusion coefficient (SDE (\ref{eq:sdepathloss})) leads to
increased achieved users' utility. This result emphasizes the
importance of a more realistic model for the power loss such as the
one of SDE (\ref{eq:sdepathloss}), as opposed to the mean power loss
model used in the conventional NUM problem formulation.

\begin{thm}
Let us consider two networks $1$, $2$, one noisier than the other,
but otherwise identical. Precisely, for
$k\in\{1,2\},~\forall\,(i,j),~\forall\,t\in [t_0,T],$
\begin{eqnarray}
dX_{ij}^k(t)=\beta_{ij}(t)\big(\gamma_{ij}(t)-X_{ij}^k(t)\big)\,
dt+\delta_{ij}^k(t)\, dW_{ij}^k(t),~
X_{ij}^k(t_0)=x_{ij,0},~\mathrm{and}~\delta_{ij}^1(t)\leq
\delta_{ij}^2(t). \label{eq:sdepathlossi}
\end{eqnarray}
\noindent If $J^*_k, ~k\in\{1,2\},$ is the optimal objective value
achieved for each network, then $J^*_1 \leq J^*_2$.
\label{thm:volatilityincrease}
\end{thm}
\begin{proof}
The solutions of the SDEs in Eq. (\ref{eq:sdepathlossi}) can be
written, $\forall (i,j)$, as
\begin{eqnarray}
X_{ij}^k(t)=X_{ij}^{det}(t)+\int_{t_0}^t
\delta_{ij}^k(s)e^{-\int_s^t
\beta_{ij}(r)dr}dW_{ij}^k(s),~k\in\{1,2\},
\label{eq:solutionpowerlossi}
\end{eqnarray}
\noindent where $X_{ij}^{det}(t)=e^{-\int_{t_0}^t
\beta_{ij}(s)ds}x_{ij0}+ \int_{t_0}^t \beta_{ij}(s)
\gamma_{ij}(s)e^{-\int_s^t \beta_{ij}(r)dr}ds$
 is the solution for a noiseless channel.
Note that
$\EE_{t_0,x_0}\big[X_{ij}^k(t)\big]=X_{ij}^{det}(t),~k\in\{1,2\}.$

 By
Proposition $3.1$ in \cite{hirsch14} we have that
$X_{ij}^1(t)\stackrel{(c)}{\leq} X_{ij}^2(t), ~\forall t \in
[t_0,T]$, 
where $\stackrel{(c)}{\leq}$ stands for partial ordering in the
convex order. That is, for every convex function $\psi$ we have
$\EE_{t_0,x_0}\left[ \psi\big(X_{ij}^1(t)\big) \right] \leq
\EE_{t_0,x_0}\left[ \psi\big(X_{ij}^2(t)\big) \right], ~\forall t
\in [t_0,T]$. 
One such convex function is Shannon's formula for the channel's
capacity $C_{ij}$, i.e. $x\rightarrow
B_{ij}\log_2\left(1+\frac{e^{Kx}P_{ij}}{N_{0}}\right)$. Hence,
\begin{eqnarray}\EE_{t_0,x_0}\left[ C_{ij}^1(t) \right] \leq
\EE_{t_0,x_0}\left[ C_{ij}^2(t) \right], ~\forall t \in [t_0,T].
\label{eq:result3}
\end{eqnarray}
In particular, let us denote by $\mathfrak{L}^k,~k\in\{1,2\}$, the
set of the deterministic $\Lambda \in \mathcal{U}_{\lambda}^F,~R \in
\mathcal{U}_{r}^{E\times (N-1)}$ that satisfy the constraints of
$P_2$ imposed on each network:
\begin{eqnarray}
\EE_{t_0,x_0}\left[\int_{t_0}^T \lambda_i^d(t)dt +\int_{t_0}^T
\sum_{j:i \in \mathcal{R}(j,d)} r_{ji}^d(t)dt\right]\leq
\EE_{t_0,x_0}\left[\int_{t_0}^T
\sum_{j\in \mathcal{R}(i,d)} r_{ij}^d(t)dt\right],~ \forall i,d, \\
\EE_{t_0,x_0}\left[\int_{t_0}^T \sum_{d} r_{ij}^d(t)dt\right]\leq
\EE_{t_0,x_0}\left[\int_{t_0}^T \zeta_{ij} C_{ij}^k(t)dt\right],
~\forall (i,j)\in \mathcal{E},~k=\{1,2\}.
\end{eqnarray}
 \noindent In view of Rel. (\ref{eq:result3}) we have that $\mathfrak{L}^1\subseteq \mathfrak{L}^2$ and
therefore, as asserted,
\begin{eqnarray}
J_1^*=\max_{\Lambda,~R  \in \mathfrak{L}^1}
\EE_{t_0,x_0}\left[\int_{t_0}^T \sum_{i,d} U_i^d(\lambda_i^d(t),t)
dt\right]\leq \max_{\Lambda,~R  \in \mathfrak{L}^2}
\EE_{t_0,x_0}\left[\int_{t_0}^T \sum_{i,d} U_i^d(\lambda_i^d(t),t)
dt\right]=J_2^*.
\end{eqnarray}
\end{proof}

\begin{remark}
Note that
 $C_{ij}\big(\delta_{ij}(.)\big)\ge
B_{ij}\EE^{A_0}\Big[\log_2\big(1+\frac{P_{ij}e^{KX_{ij}(t)}}{N_0}\big)\Big]
\ge \mathbf{P}(A_0)
B_{ij}\log_2\frac{P_{ij}}{N_0}+\frac{B_{ij}K}{\log
2}\EE^{A_0}\big[X_{ij}(t)\big]$, where $A_0=\big\{X_{ij}(t)\ge
\EE[X_{ij}(t)]\big\}$. Since $X_{ij}(t)$ are gaussian,
$\mathbf{P}(A_0)=1/2$ by symmetry, and $\EE^{A_0}[X_{ij}(t)]=
\frac{1}{2}\EE[X_{ij}(t)]+\sqrt{\frac{V(X_{ij}(t))}{2\pi}}$ where
$V(X_{ij}(t))$ is the variance of $X_{ij}(t)$ given by
$\int_{s}^{t}\delta_{ij}^2(r)\exp(-2\int_{r}^{t}\beta_{ij}(q)dq)dr$.
Therefore, the link capacity may take arbitrarily large values if
$\int_s^T \delta_{ij}^2(t)\,dt$ is sufficiently large.
\end{remark}

\section{Numerical Results} \label{sec:simulation} In this section,
we present and discuss indicative numerical results evaluating the
proposed schemes focusing on the case of orthogonal access to the
medium and LTF. After describing the evaluation setting and some
general observations, we illustrate numerically Theorem
\ref{thm:volatilityincrease} and the behavior of the proposed
framework in case of congestion control and routing (i.e.
optimization in transport and network layers). The latter is then
compared with the behavior of the joint routing, scheduling,
congestion and power control scheme (i.e. cross-layer optimization).
In addition, we examine the behavior of the proposed framework in
the case of time varying utilities and the possibility of applying
the solution procedure of the dual problem $D_2$, during the network
operation (online deployment).

We consider a wireless multihop network of $N=16$ nodes, forming a
$4\times 4$ grid topology. For ease of presentation, we consider
that the parameters of the SDE (\ref{eq:sdepathloss}),
$\gamma_{ij}(t),\beta_{ij}(t),\delta_{ij}(t)$, as well as the
initial states $X_{ij}(s)=x$ are identical for every link $(i,j)$. 
 Here, we consider
$M=200$ paths, $N_{0}=0.1W$, $s=0$. Furthermore, $\forall\,
(i,j),\tau_b$, $B_{ij}=10^6 Hz$, $\gamma_{ij}(\tau_ b)=\gamma
\bigg(1+0.15 e^{(-2 \frac{\tau_b}{n})} \sin(10 \pi
\frac{\tau_b}{n})\bigg)$ \cite{olama10}, $X_{ij}(0)=\gamma$,
$n=500$, and $\beta_{ij}(\tau_b)=100$,
$\gamma=70 dB$ unless differently mentioned. Moreover,
$\delta_{ij}(\tau_b)=\delta,\forall\, (i,j),\tau_b,$ where $\delta$
will be tuned in each numerical experiment. Note that the value
$n=500$ determines the sampling rate for computing the Riemann sums
that approximate the integrals e.g. in Eqs. (\ref{eq:lagrange2b}),
(\ref{eq:lagrange2c}) and it is chosen so that the Riemann sum is
close to the corresponding integral value, while simultaneously
being small enough for trading-off the cost of sampling. Low
sampling rate impacts performance since the Riemann sum does not
converge to the actual value of the corresponding integral. On the
contrary high sampling rate may induce extra cost without offering
significant improvement in the Riemann
sum's accuracy in approximating the corresponding integral. 
The periodic behavior of the LTF wireless channel parameters may be
due to an absorbing obstacle intervening periodically between the
transmitter and the receiver (as in an example of
\cite{goldsmith05}).
In the absence of scheduling optimization, for each link $(i,j)$ the
value $\zeta_{ij}$ is computed based on scheduling all the maximal
independent sets of the network topology for equal percentage of
time. Regarding the traffic model, each node chooses a random
destination among its non-physically connected nodes, and becomes
the source for this destination. Initially, we consider logarithmic
utilities, i.e. $U_i^d(\lambda_i^d)=\log (\lambda_i^d)$, a common
choice in the literature to model elastic traffic
\cite{georgiadis06}.

Each numerical experiment for the determination of the optimal
control variables and the optimal Lagrange multipliers runs until
convergence is ensured and specifically until the sum of the changes
between consecutive values of the Lagrange multipliers over the
preceding eight iterations is less than $0.001$. The learning rate
is chosen as $\kappa(\eta)=\frac{A'}{\eta},~\forall\,\eta$
\cite{bertsekas09}, where $A'=0.1$ so that convergence is allowed in
a reasonable time duration with respect to the chosen convergence
criterion. Specifically, by decreasing $A'$ by one or more orders of
magnitude our criterion of convergence is satisfied too soon,
impeding the subgradient algorithm approach to the global optimal
values. On the other hand, increasing $A'$ or using constant values
of $\kappa(\eta),~\forall\,\eta$  do not lead to convergence within
a reasonable time interval.

In order to obtain an intuition regarding the wireless channel's
stochastic behavior, Figs. \ref{fig:capacitypaths},
\ref{fig:capacitypaths1} show typical sample paths of the solution
to the SDE (\ref{eq:sdepathloss}) for $\delta=25$ and $\delta=50$,
respectively. As expected, we observe larger deviations of the power
loss from $\gamma_{ij}(\cdot)$ for the higher choice of $\delta$.
Consequently, the capacity may achieve higher values due to random
fluctuations in this case. Note that although the curve of
$\gamma_{ij}(\cdot)$ tends to weaken to a line as time increases, it
fluctuates considerably for the duration of the network's operation
(transient state). Fig. \ref{fig:capacitypaths3} shows typical
sample paths of the solution to the SDE (\ref{eq:sdepathloss}) for
time varying speed of adjustment $\beta_{ij}(\tau_b)$. When
$\beta_{ij}(\tau_b)$ attains low values (i.e. $\tau_b\leq166$) the
power loss diverges more from its attraction curve
$\gamma_{ij}(\tau_b)$, while the attraction to $\gamma_{ij}(\tau_b)$
becomes faster when $\beta_{ij}(\tau_b)=500$ (i.e. $\tau_b>333$).
\begin{figure}[t]
    \centering
    \subfigure[$\delta=25$.]
      {\includegraphics[width=2in]{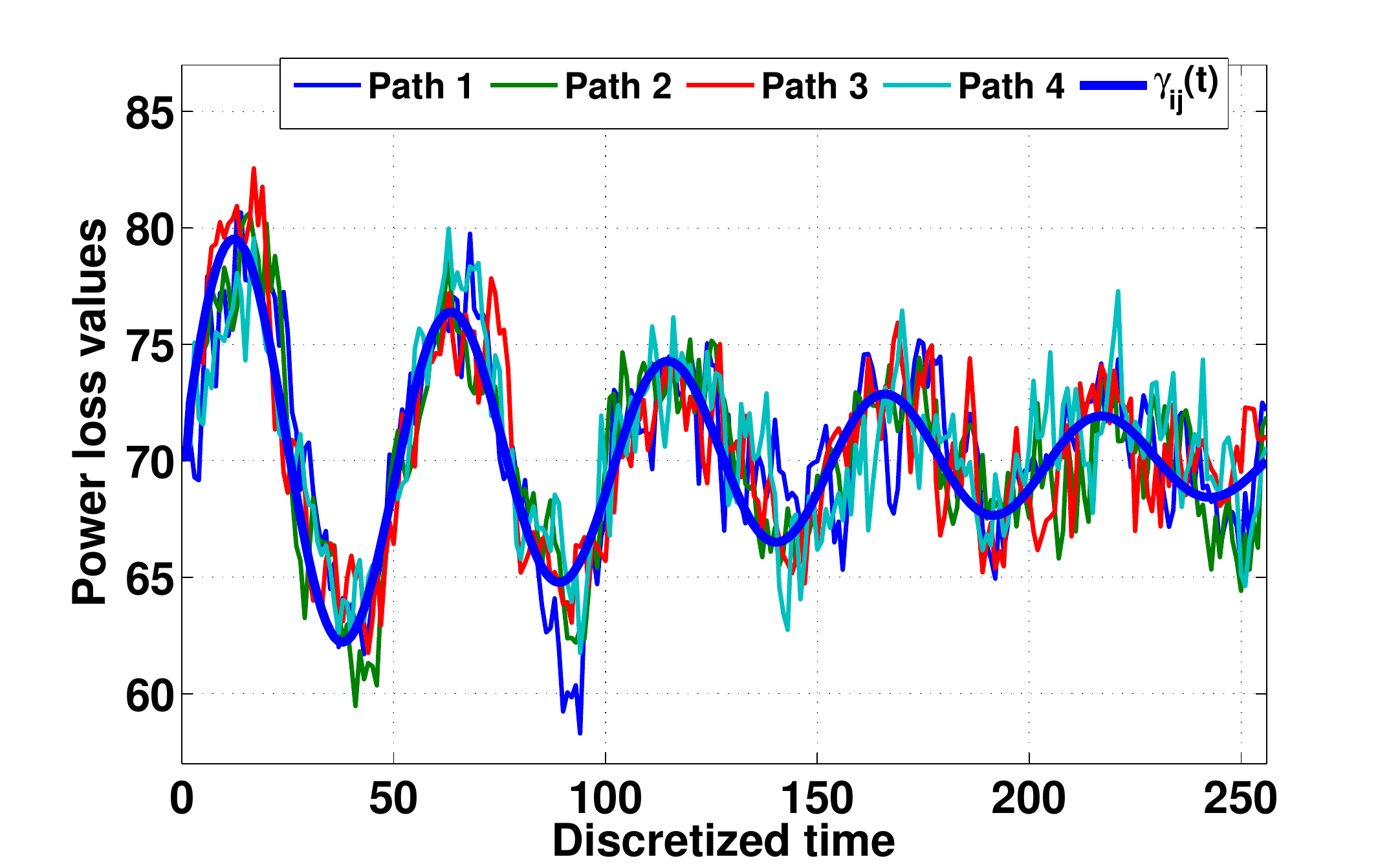}
    \label{fig:capacitypaths}}
     \subfigure[$\delta=50$.]
      {\includegraphics[width=2in]{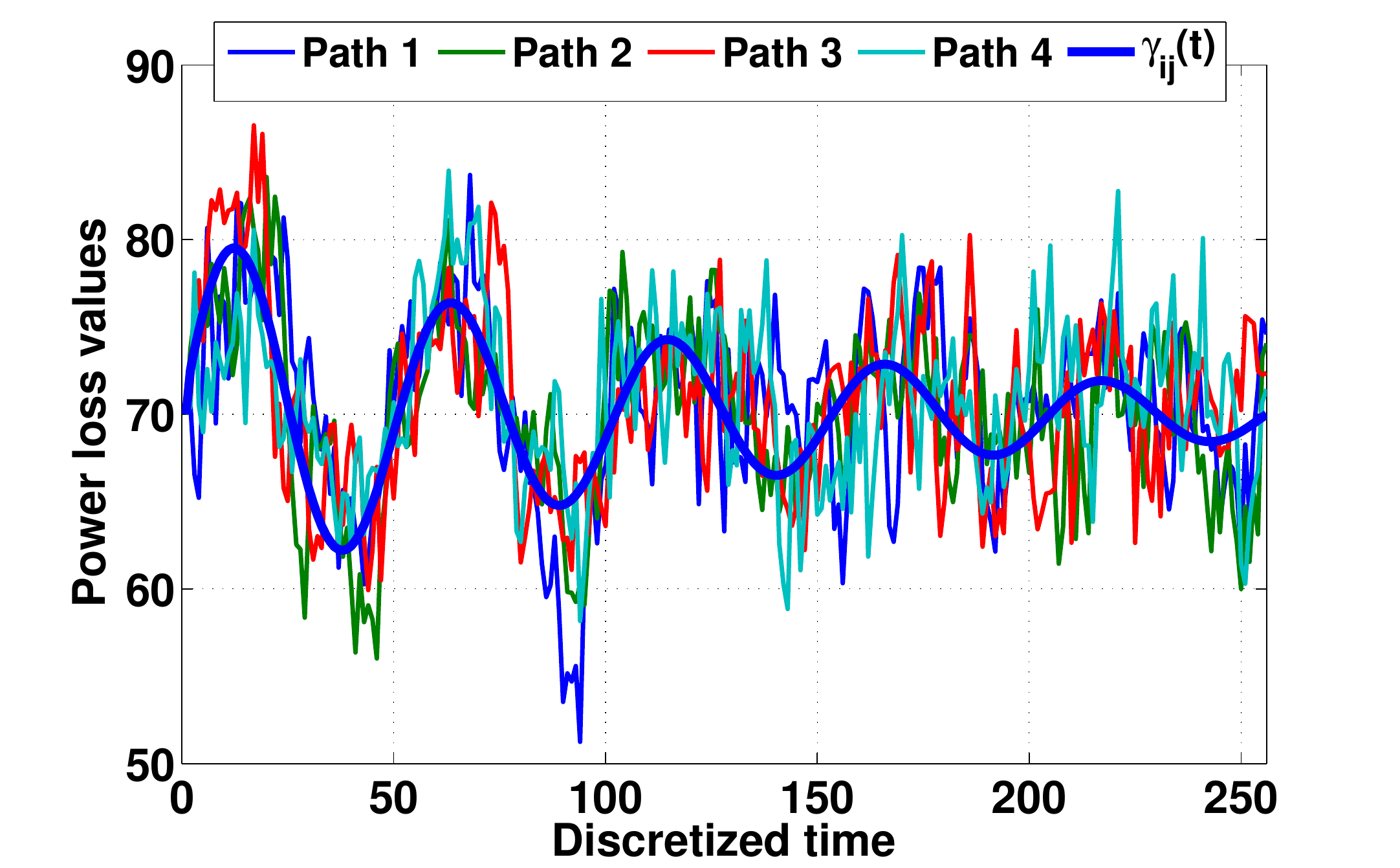}
    \label{fig:capacitypaths1}}
    \subfigure[$\delta=20$, $\beta=10$ for $\tau_b\leq166$, $\beta=100$ for $166<\tau_b<333$ else $\beta=500$.]
      {\includegraphics[width=2in]{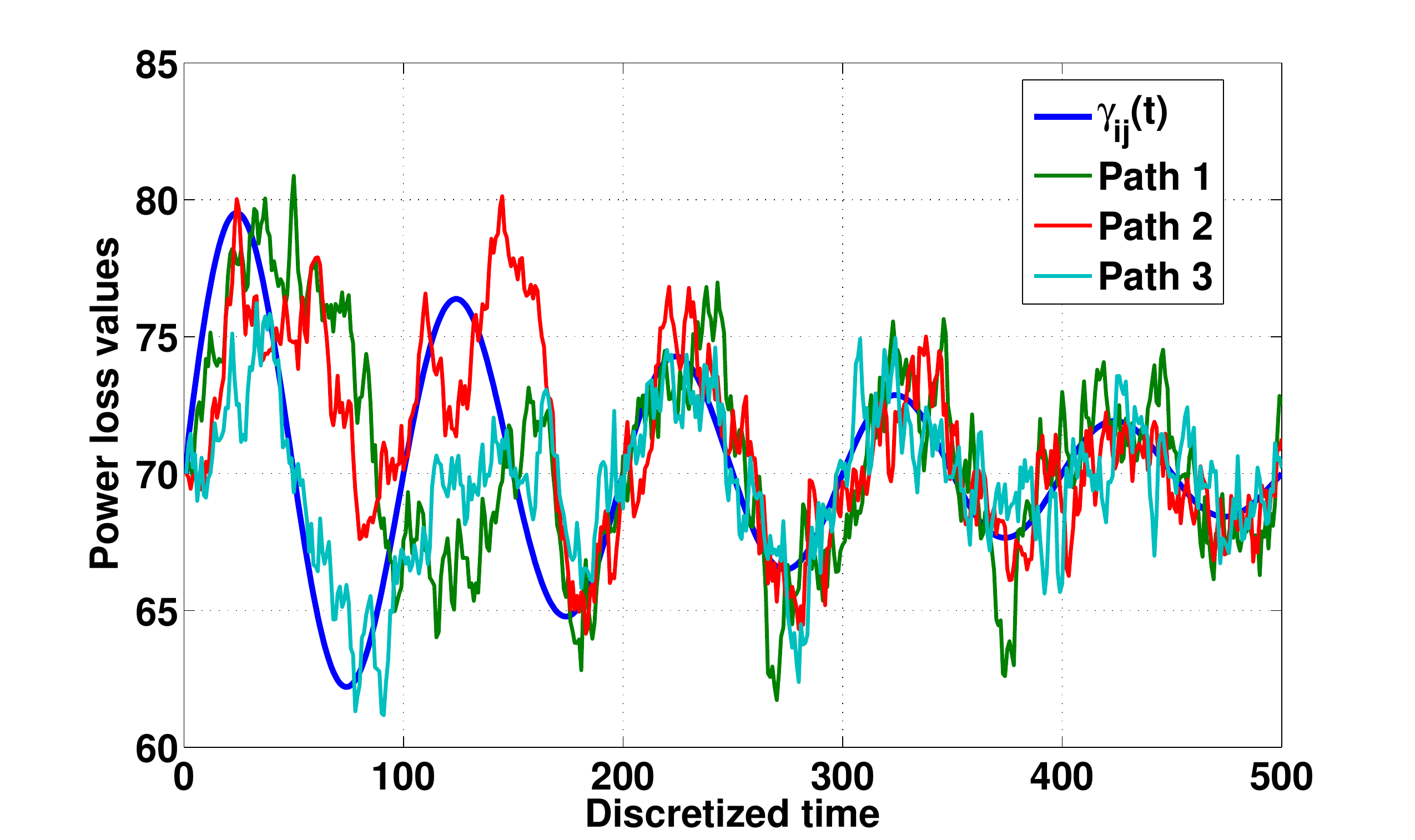}
    \label{fig:capacitypaths3}}
    \caption{Wireless channel's LTF stochastic behavior. Paths for the power loss (dB) modeled via SDE.}
\end{figure}
\subsection{Congestion Control \& Routing}
In the sequel, we examine the case of applying only congestion
control and routing, assuming the transmitter of each link $(i,j)$
has a constant power of $P_{ij}(t)=2 W$, $\forall~t\in [0,T]$.
\begin{figure}[t]
    \centering
    \subfigure[Optimal source rates (bits/sec) vs. $\delta$.]
      {\includegraphics[width=2in]{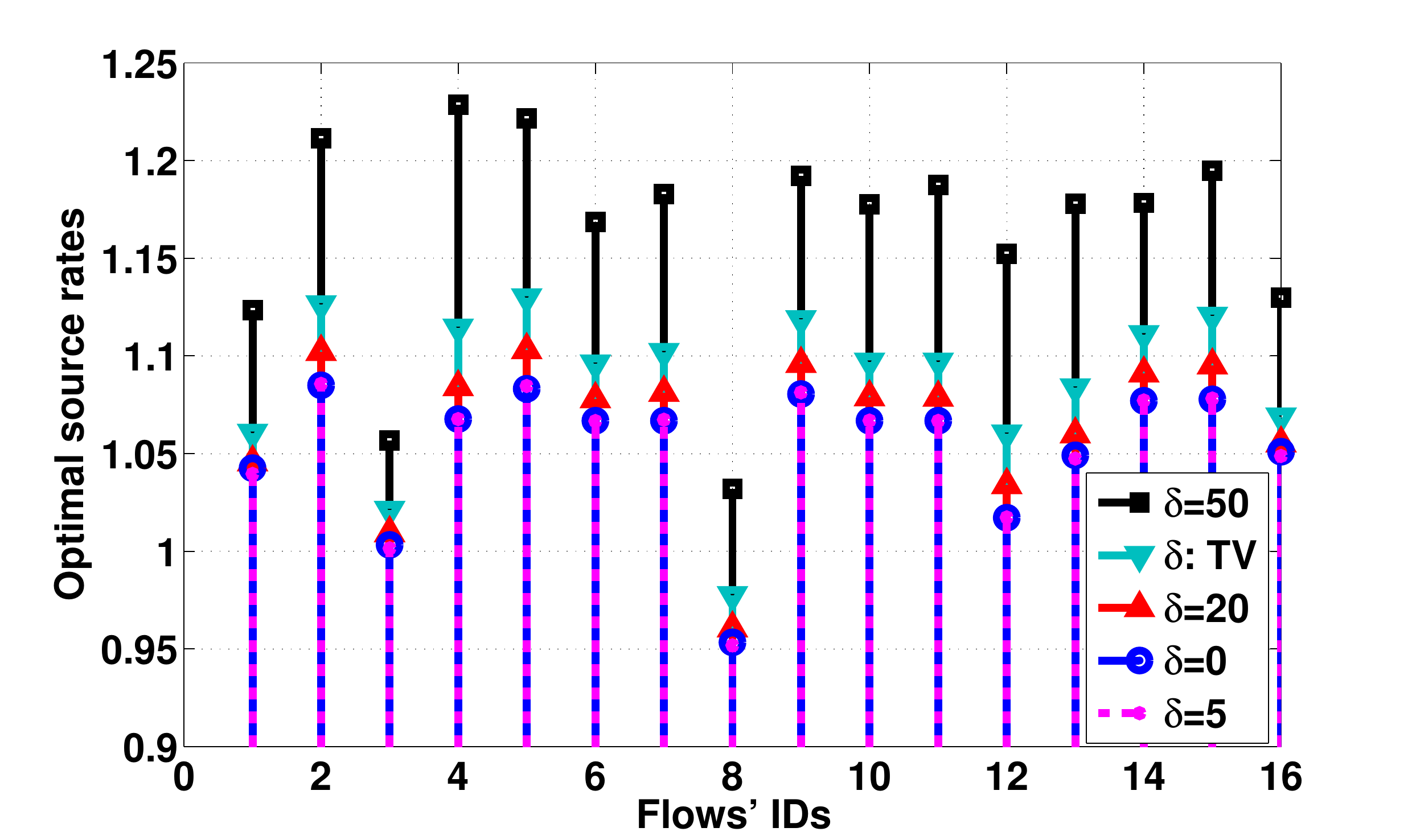}
    \label{fig:routingcongestiondelta}}
    \subfigure[Convergence of the Lagrange multipliers $\{\mu_i^d\}_{\forall i,d}$.]
      {\includegraphics[width=2in]{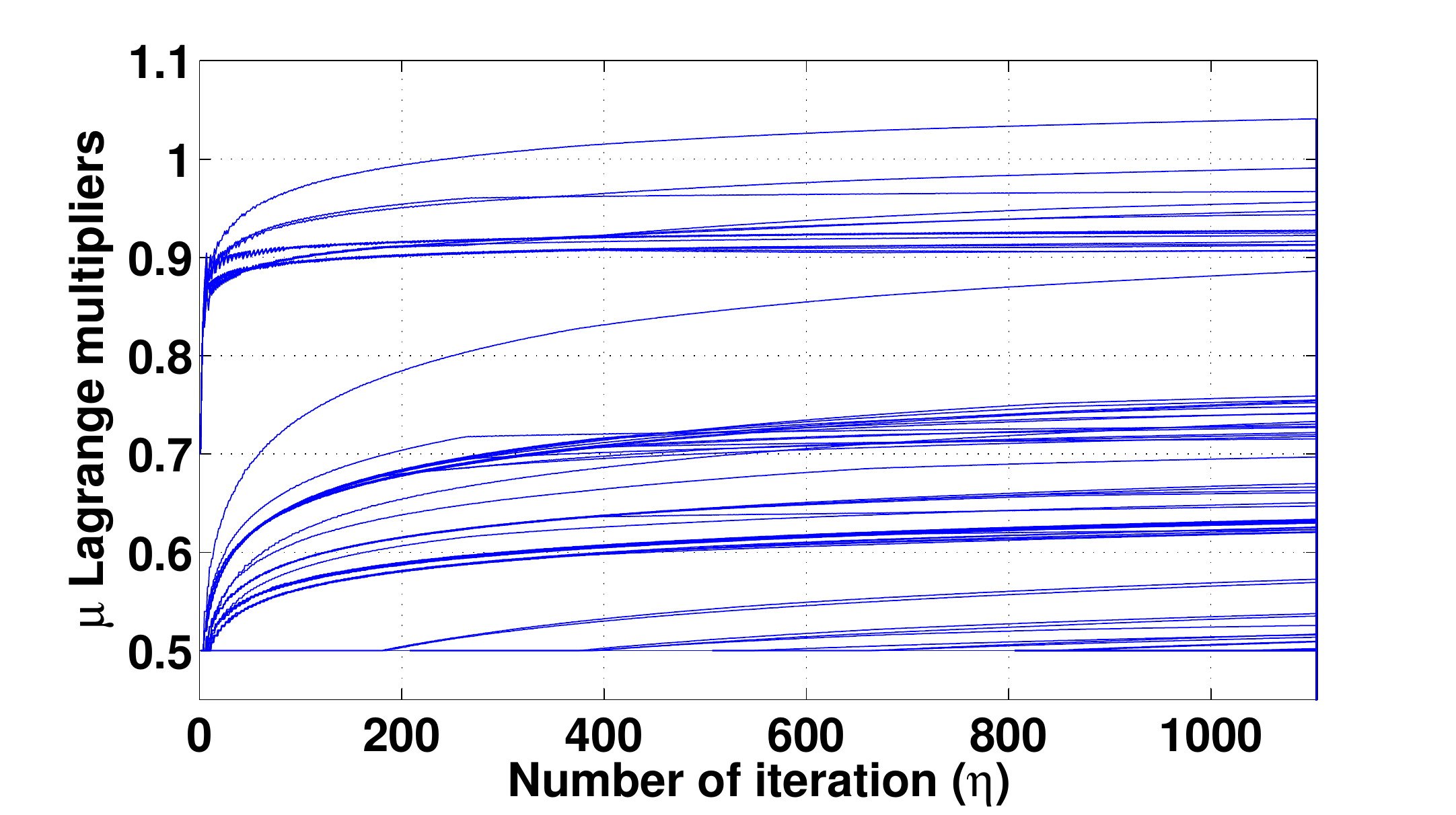}
    \label{fig:convM}}
 \subfigure[Convergence of the Lagrange multipliers $\{l_{ij}\}_{\forall (i,j)}$.]
      {\includegraphics[width=2in]{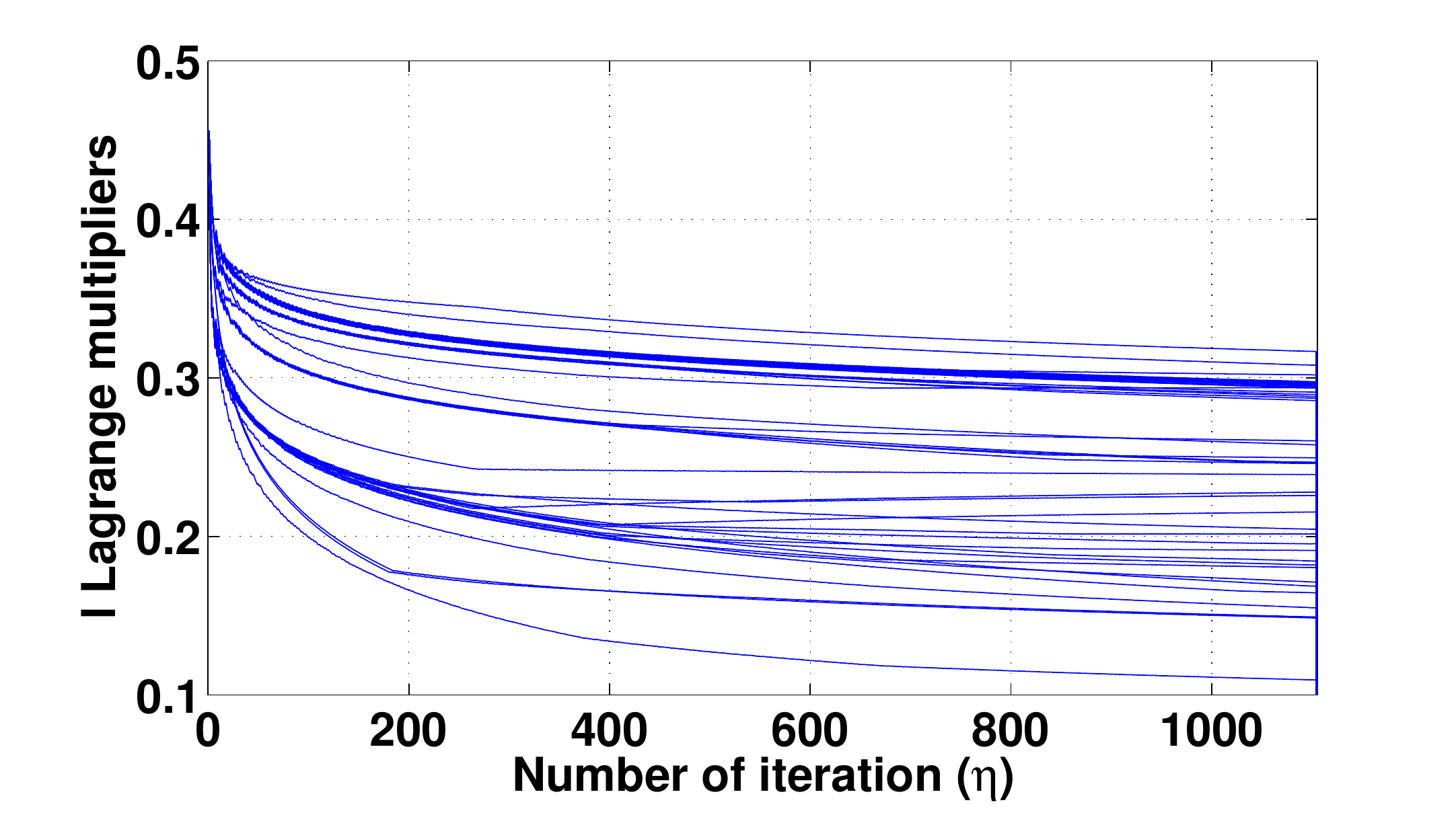}
    \label{fig:convL}}
    \caption{Congestion control \& routing: Optimal source rates and convergence of the Lagrange multipliers.}
\end{figure}
Fig. \ref{fig:routingcongestiondelta} shows the optimal source rates
(i.e. after convergence since they change continuously with respect
to the optimal Lagrange multipliers) for different choices of the
diffusion coefficient $\delta$. It can be observed that as $\delta$
increases, random fluctuations to higher capacity values are
exploited to offer increased optimal source rates, hence verifying
numerically the statement of Theorem \ref{thm:volatilityincrease}.
It is important to note here that the noise level $\delta$ does not
affect the mean power loss, as indicated in the proof of Theorem
\ref{thm:volatilityincrease}. In other words, if the mean power loss
is used to determine capacity, higher capacity values due to random
fluctuations cannot be tracked and exploited for increasing the
source rates. Time-varying $\delta$ (TV) is also applied, and
specifically $\delta=15 sin(10 \pi d/n)+35$ i.e. taking values
between $\delta=20$ and $\delta=50$, thus leading to optimal source
rates in between the ones corresponding to these two values of
$\delta$. For benchmarking purposes, in Fig.
\ref{fig:routingcongestiondelta} we have added the case of
$\delta=0$, which corresponds to time-varying but deterministic
channels. We observe that the optimal source rates achieved under
deterministic wireless channels are the lowest (approximately the
same as in the case of a low value of noise, i.e. the case that
$\delta=5$), indicating the improvement in system's performance when
accounting for randomness.

Figs. \ref{fig:convM}, \ref{fig:convL} depict the behavior of the
proposed scheme considering convergence to the optimal Lagrange
multipliers. We also study the impact of $T$ on the time to
convergence and the achieved source rates. The value of $n$ is
adapted for each $T$ as it is shown in Fig. \ref{fig:Limpact}. We
observe that as the duration of the network's operation, $T$,
increases, the time to convergence also increases (Fig.
\ref{fig:Limpact}), while the achieved arrival source rates decrease
for all flows (Fig. \ref{fig:Timpact}). Finally, the behavior of the
proposed algorithm in case of time varying $\beta_{ij}(t)$ with
respect to the optimal source rates is shown in Fig.
\ref{fig:betaimpact}. We observe that when $\beta_{ij}(t)$ is high
at the beginning (close to time $s=0$) the optimal rates are higher
than when $\beta_{ij}(t)$ is initially low and increases later in
time.
\begin{figure}[t]
    \centering
\subfigure[Impact of $T$ on the time to convergence (iterations).]
      {\includegraphics[width=2in]{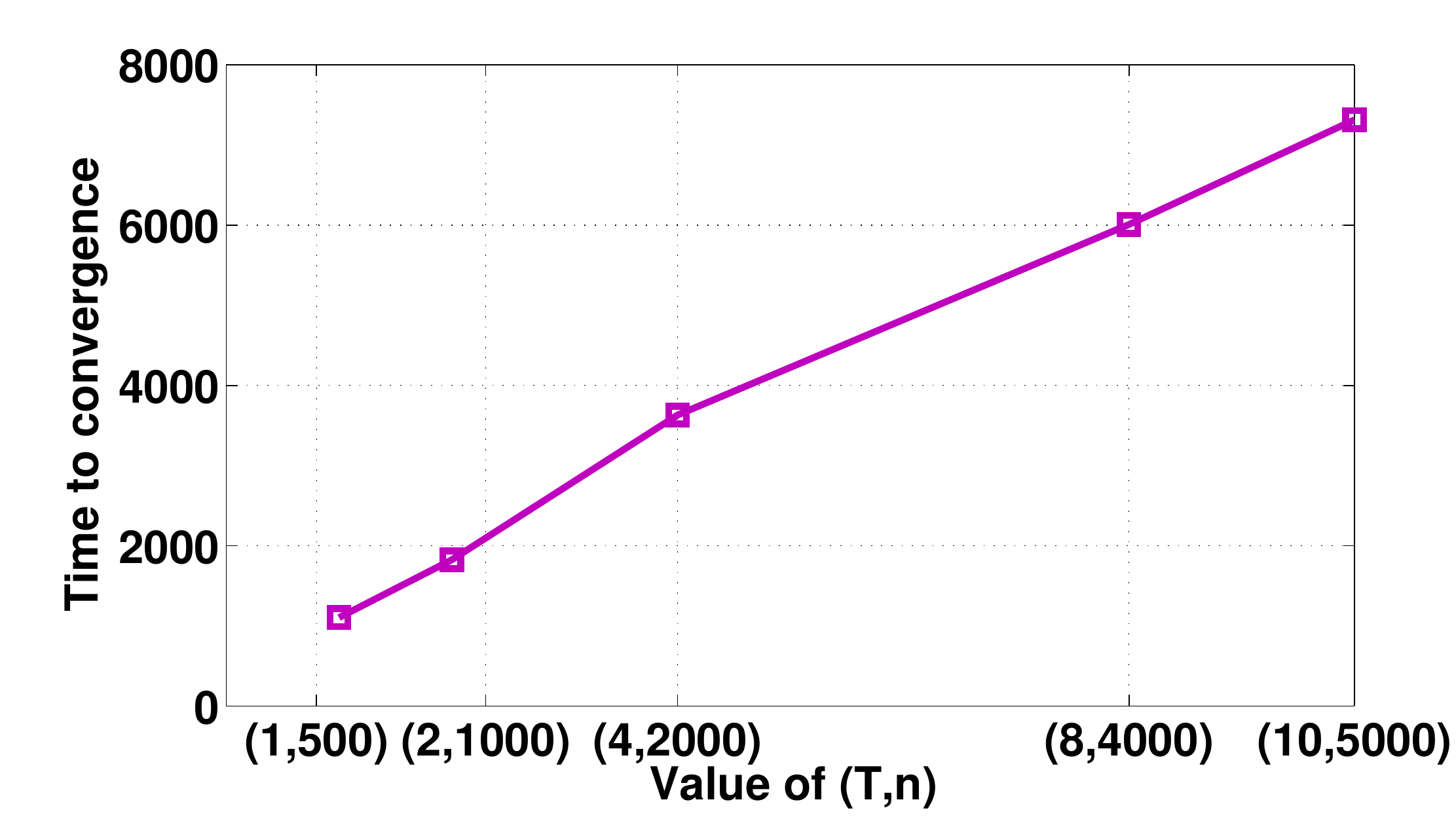}
    \label{fig:Limpact}}
    \subfigure[Impact of $T$ on the optimal flows' rates (bits/sec).]
      {\includegraphics[width=2in]{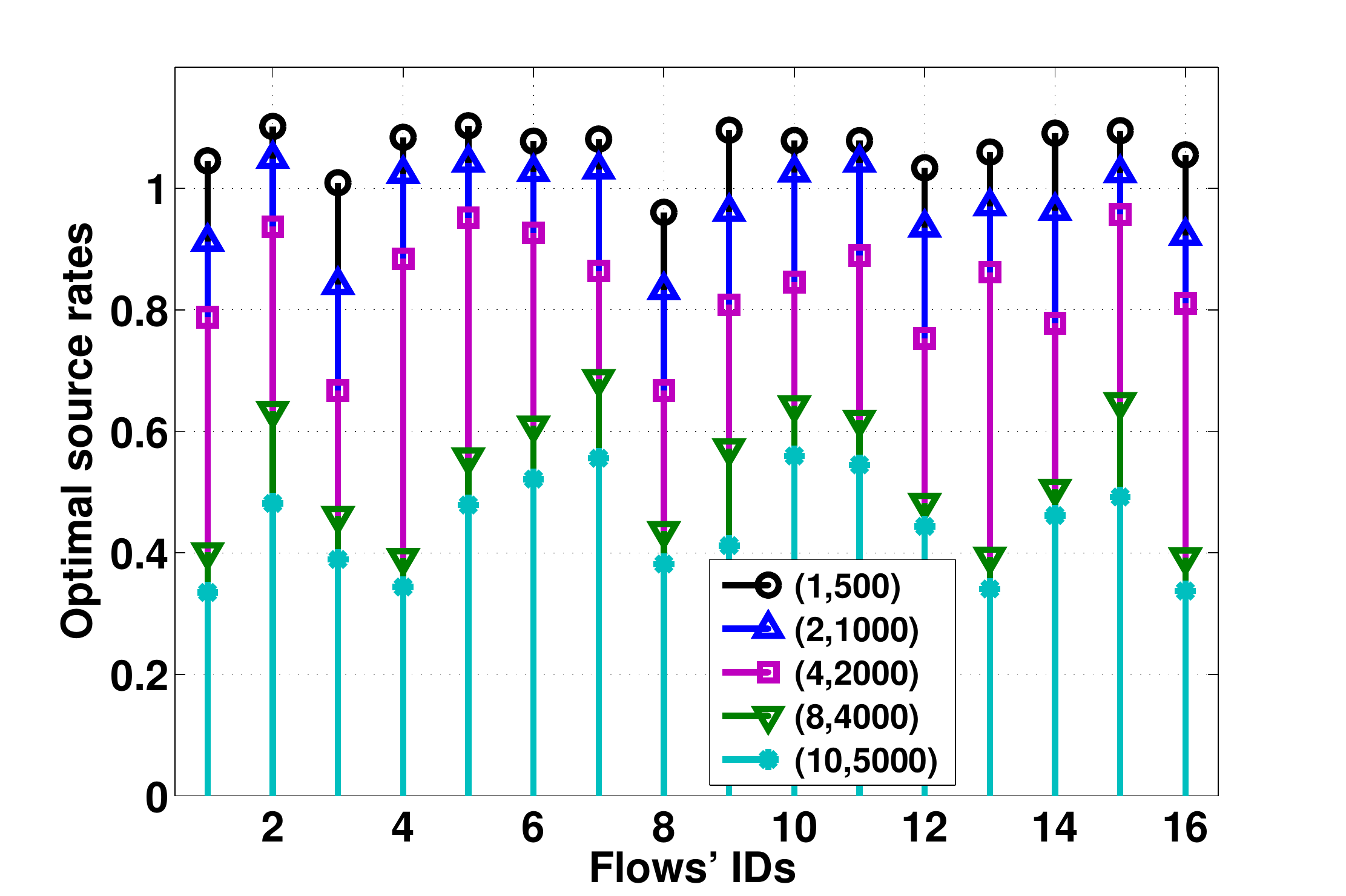}
    \label{fig:Timpact}}
    \subfigure[Impact of time varying $\beta_{ij}(t)$ of the SDE (\ref{eq:sdepathloss}) on the optimal flows' rates (bits/sec).]
      {\includegraphics[width=2.2in]{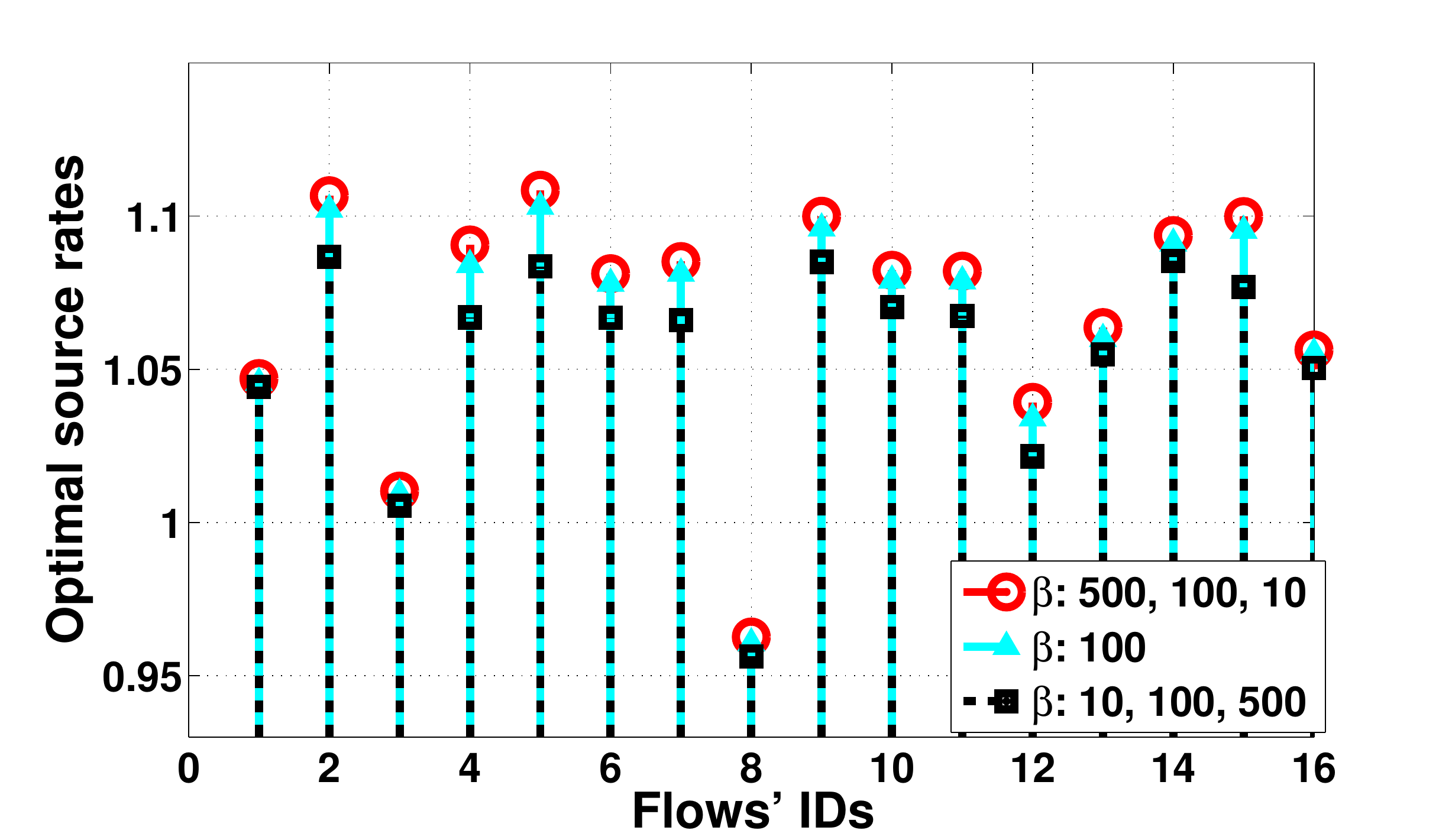}
    \label{fig:betaimpact}}
    \caption{Congestion control \& routing: Study of the impact of $T$ and time-varying $\beta$ on convergence. In subfigure (c), black corresponds to $\beta=10$ for $\tau_b\leq166$, $\beta=100$ for $166<\tau_b<333$ else $\beta=500$, red corresponds to $\beta=500$ for $\tau_b\leq166$, $\beta=100$ for $166<\tau_b<333$ else $\beta=10$, and cyan corresponds to $\beta=100$.}
\end{figure}

\subsection{Joint Scheduling, Routing, Congestion \&
Power Control Scheme} In order to evaluate the joint scheduling,
routing, congestion and power control scheme, we assume that each
link can vary its transmission power between $1W$ and
$3W=P_{i,\max}, ~\forall i$. The network topology and traffic along
with the rest of the parameters remain the same as in the previous
experiments. Fig. \ref{fig:capacitypower} depicts the optimal
transmission power at each repetition of channel state's sampling
derived from Eq. (\ref{eq:p3solution})
and the corresponding path of the link capacity. It is observed that
the optimal power increases when power loss decreases (thus capacity
increases) and attains low values for high values of power loss, as
expected from the analysis of Section \ref{sec:orthogonal}.
Therefore, the transmission power increases only when there is an
opportunity for an important capacity improvement due to random dips
of power loss. On the contrary, transmission power is not wasted
when the stochastic power loss does not support capacity increase.
Fig. \ref{fig:powerroutingcongestiondelta} shows the optimal source
rates (i.e. after convergence) for different choices of the
diffusion coefficient $\delta$. As in the previous case (Fig.
\ref{fig:routingcongestiondelta}), it can be observed that as
$\delta$ increases the optimal source rates increase for all flows.
This is beyond the scope of Theorem \ref{thm:volatilityincrease},
that does not account for power control and scheduling in the
optimization problem. Fig. \ref{fig:powerroutingcongestiondelta}
also includes the optimal source rates when $\delta$ is time-varying
(TV) and when channels are deterministic ($\delta=0$), leading to
similar conclusions as in the case of congestion control and routing
(Fig. \ref{fig:routingcongestiondelta}). Figs. \ref{fig:convMP},
\ref{fig:convLP} depict the convergence to the optimal Lagrange
multipliers.

 Fig.
\ref{fig:ratescomparison} shows that the joint scheduling, routing,
congestion and power control scheme improves the optimal source
rates for all flows, compared to the congestion control and routing
scheme. As also stated in \cite{goldsmith05}, applying power control
at the transmitter is like having channel state information at both
the transmitter and the receiver which improves the network capacity
compared with the case of constant power which is equivalent to
having channel state information at the receiver's side only.
Furthermore, we study the expected per link transmission power over
the entire time interval $[0,T]$ for the numerical experiments of
Fig. \ref{fig:powercongroutsch}. It can be easily computed that it
is equal to $1.0399W$ for $\delta=50$, $1.2658W$ for $\delta=20$,
$1.5634W$ for $\delta=5$, $1.6249W$ for $\delta=0$. On the one hand
these values are much smaller than the constant power of $2W$ used
to evaluate the routing and congestion control scheme (Fig.
\ref{fig:routingcongestiondelta}), thus achieving a more ``green"
network operation in addition to throughput improvement (Figs.
\ref{fig:powerroutingcongestiondelta}, \ref{fig:ratescomparison})
leading to energy efficiency. On the other hand, we observe that as
$\delta$ increases, the expected per link transmission power
decreases, indicating the importance of taking randomness into
account in operating wireless channels with energy efficiency.

\begin{figure}[t]
    \centering
    \subfigure[Capacity (bits/sec) \& optimal transmission power paths (W) ($\delta=20$).]
      {\includegraphics[width=2in]{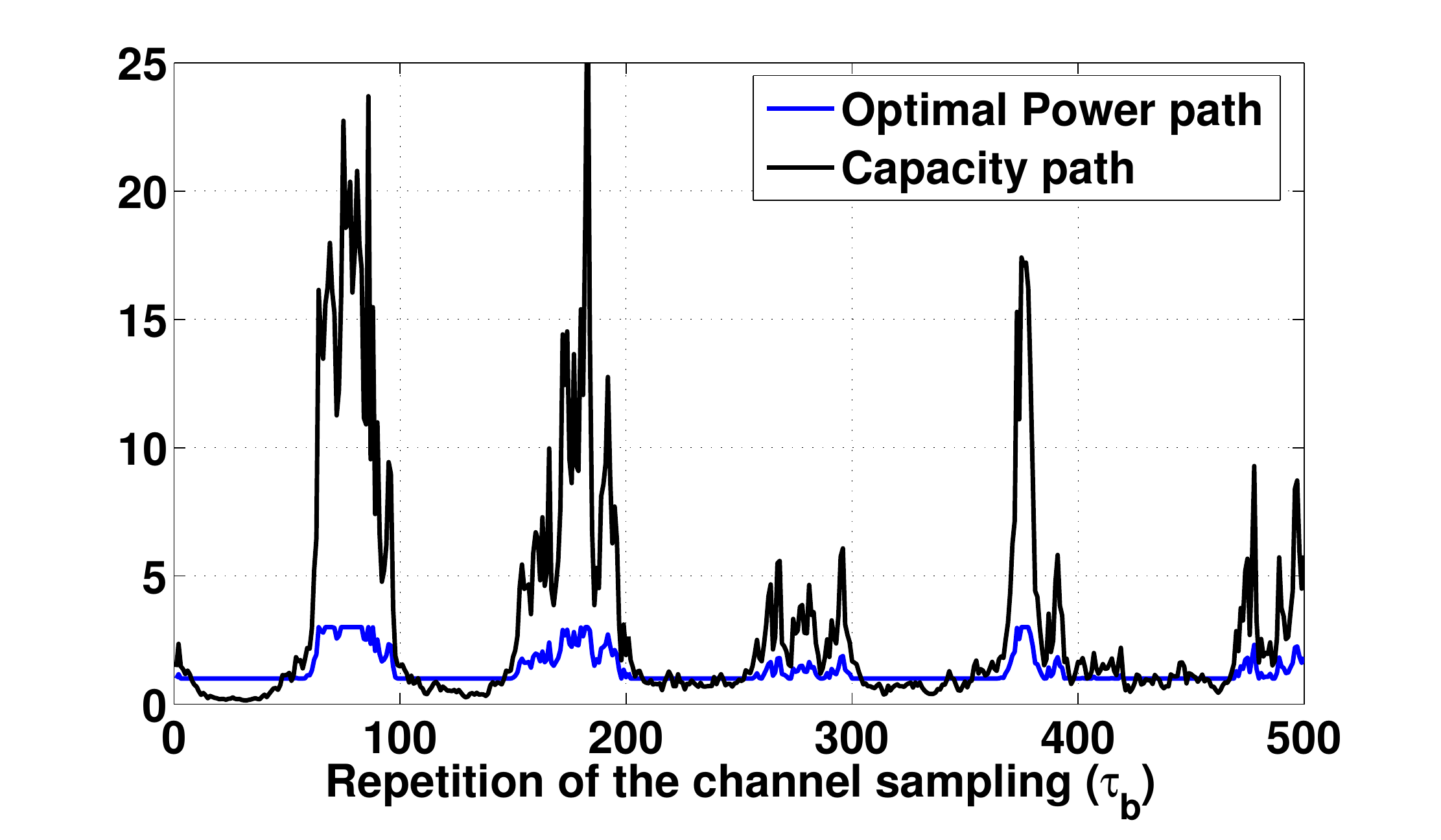}
    \label{fig:capacitypower}}
 \subfigure[Optimal source rates (bits/sec) vs. $\delta$.]
      {\includegraphics[width=2in]{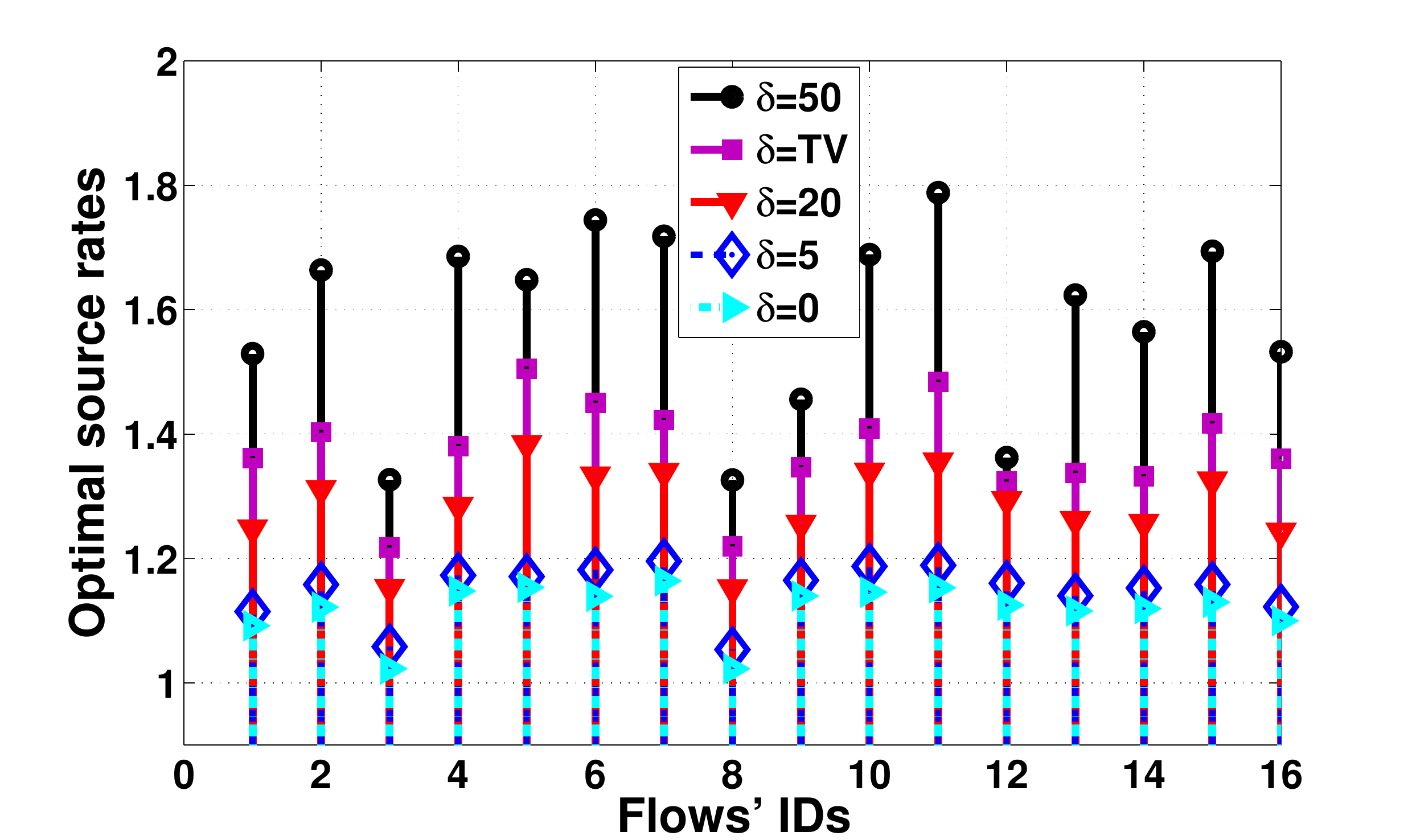}
    \label{fig:powerroutingcongestiondelta}}
     \subfigure[Comparison of optimal source rates (bits/sec) with and without scheduling and power control ($\delta=20$).]
      {\includegraphics[width=2in]{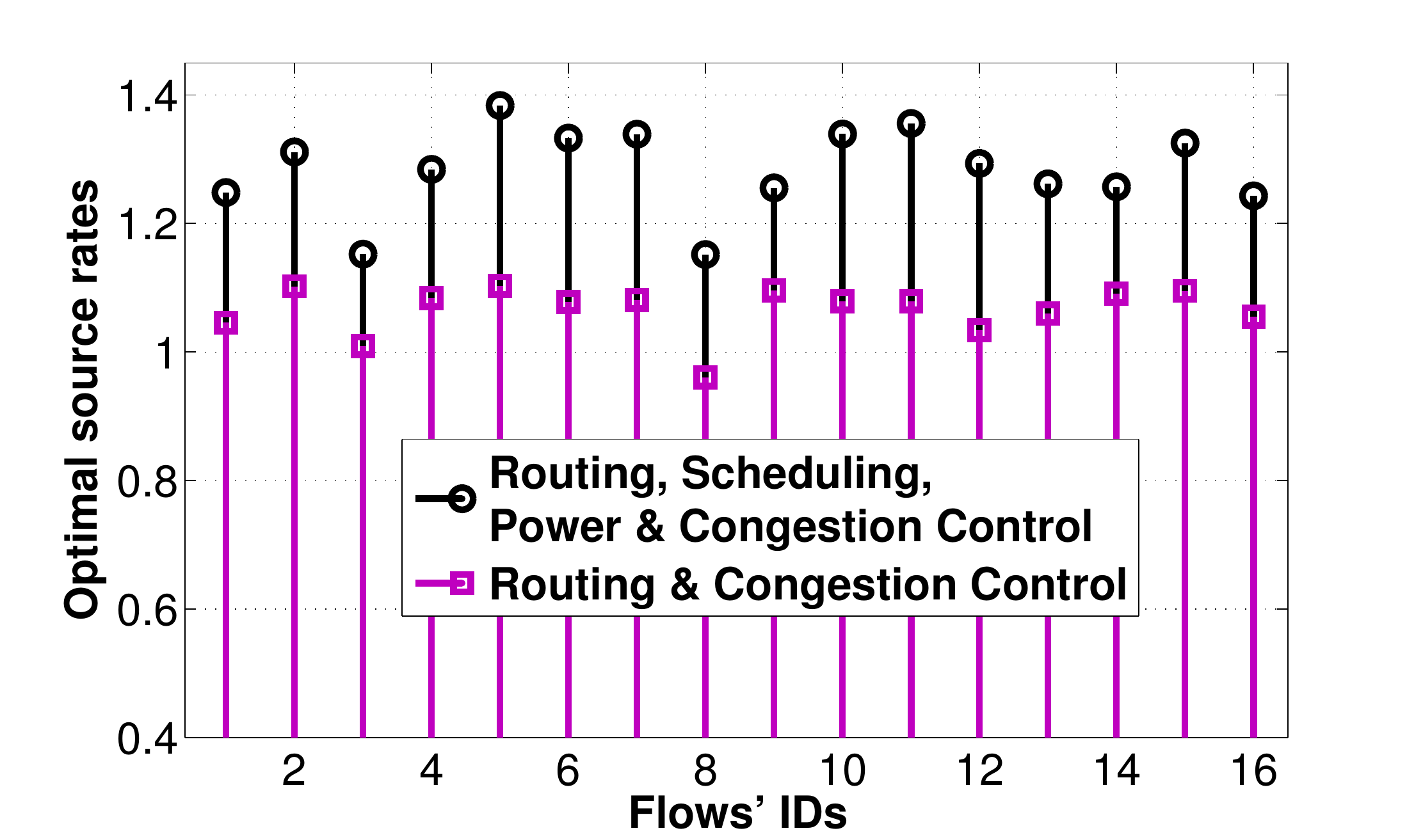}
    \label{fig:ratescomparison}}
    \centering
    \subfigure[Convergence of the Lagrange multipliers $\{\mu_i^d\}_{\forall i,d}$.]
      {\includegraphics[width=2in]{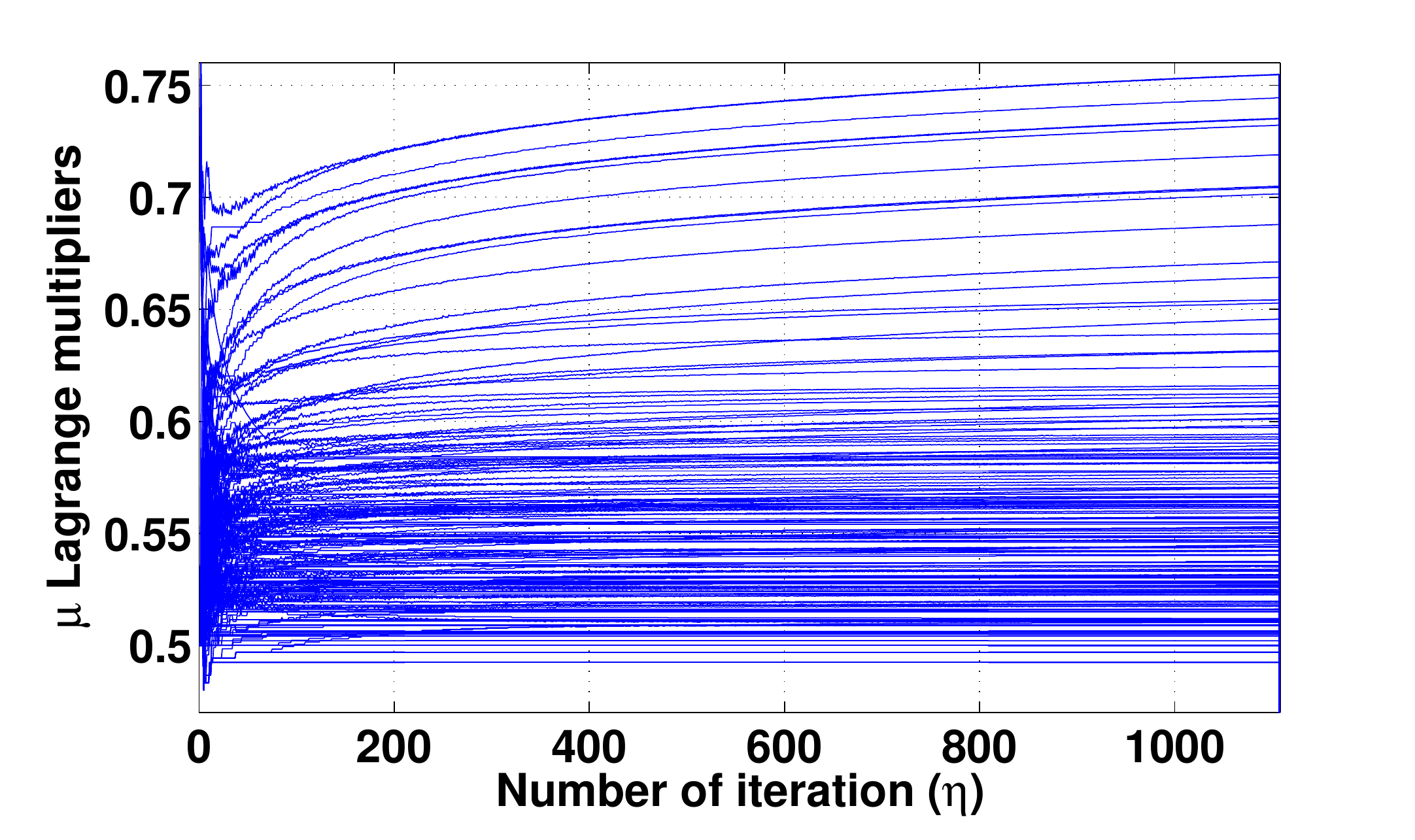}
    \label{fig:convMP}}
 \subfigure[Convergence of the Lagrange multipliers $\{l_{ij}\}_{\forall (i,j)}$.]
      {\includegraphics[width=2in]{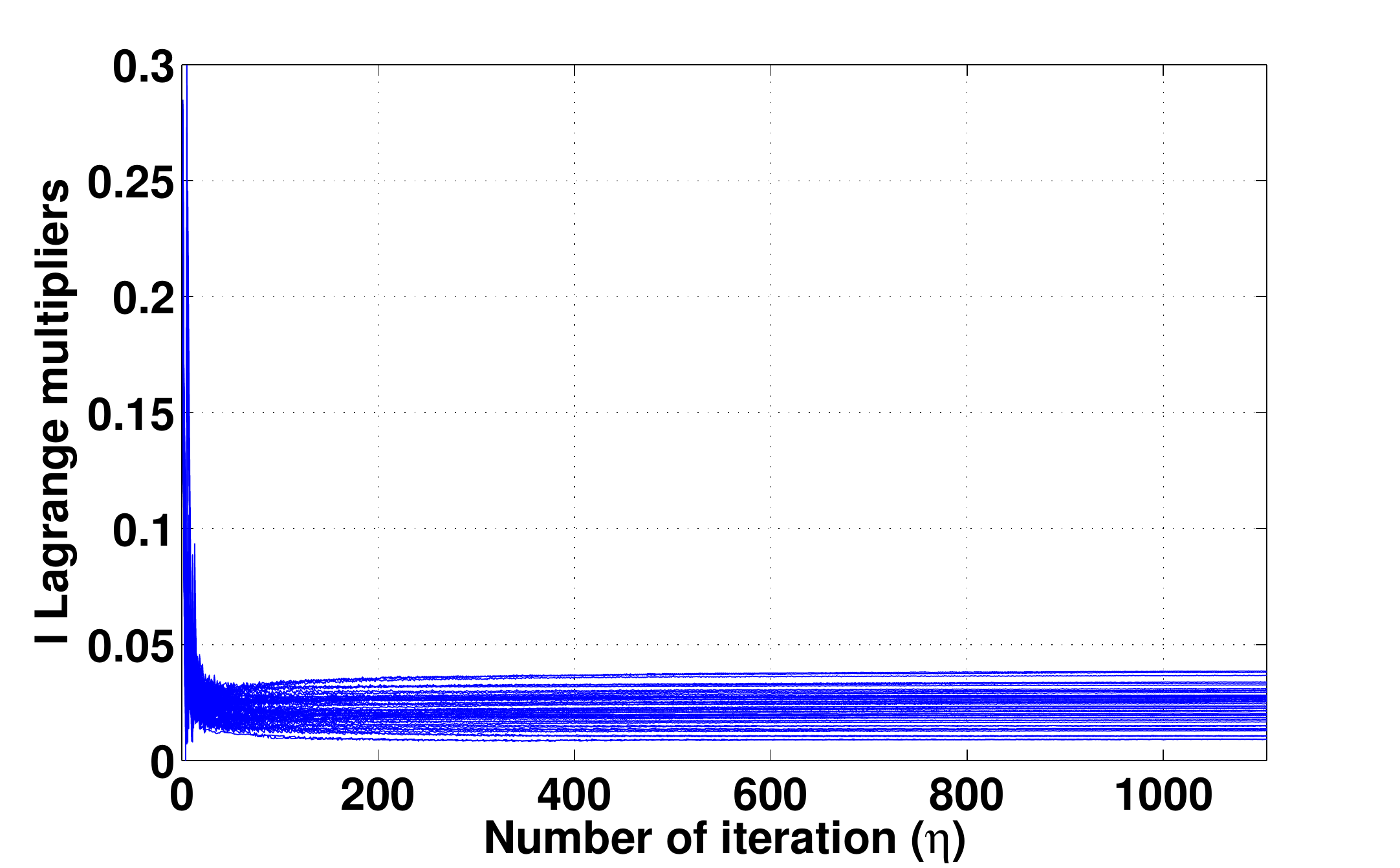}
    \label{fig:convLP}}
        \caption{Routing, Scheduling, Congestion \& Power control.}\label{fig:powercongroutsch}
\end{figure}

\subsection{Time Varying Utilities \& Online Deployment}
Finally, we study the case of time-varying utilities when applying
routing and congestion control. The utilities take the form
$U_i^d(\lambda_i^d)=\frac{\log
(\lambda_i^d)}{t},~t\in[s,T],~s>0~\forall i,d:i \in S_r(d)$, i.e.,
they decrease with time modeling the decreasing willingness of users
to produce high data amounts when approaching the end of the
network's operation. The optimal function to which the source rates
converge over $[s,T]$ is depicted in Fig.
\ref{fig:RoutingCongestionTimeVarying}. We observe that as time
increases, it dominates over the Lagrange multiplier for the
determination of the source rates (Eq. (\ref{eq:lamda})).

At this point we will make another interesting observation regarding
the online application of the proposed approach (Section
\ref{sec:orthogonal}) during the network's operation which is
initially discussed in Section \ref{sec:nonortho}. In order to
obtain an online algorithm for the network control, i.e. during the
network operation, we may consider that the network decisions are
taken at $\tau_b$ times when the channel is sampled,
while also considering that $\tau_b\equiv \eta$. Then, we should
consider the optimal control values in the interval $(\tau_b,T]$ as
``predicted" and the ones in the interval $[s,\tau_b)$ as
``corrections". However, we should study under what conditions
convergence is achieved (in practice) early enough (for small
$\tau_b$) so that optimality with respect to the achieved value of
$P_2$ is not affected. In Figs. \ref{fig:ntimeconv1},
\ref{fig:ntimeconv2}, it is shown that when increasing $n$, the time
for convergence (according to the imposed criterion) of the proposed
scheme is not significantly affected for time invariant utilities
while it is affected in a concave manner for time varying utilities.
As a result, if $n$ is large enough, the decisions taken at times
$\tau_b$ will converge fast enough compared to the whole duration
$T$. Therefore, the online application of the proposed approach
during the network operation will be suboptimal only at the
beginning barely affecting the optimal objective value of
$P_2$. 

\begin{figure}[t]
    \centering
    \subfigure[Optimal functions (bits/sec) to which the source rates converge ($\delta=20$) for time varying utilities.]
      {\includegraphics[width=2in]{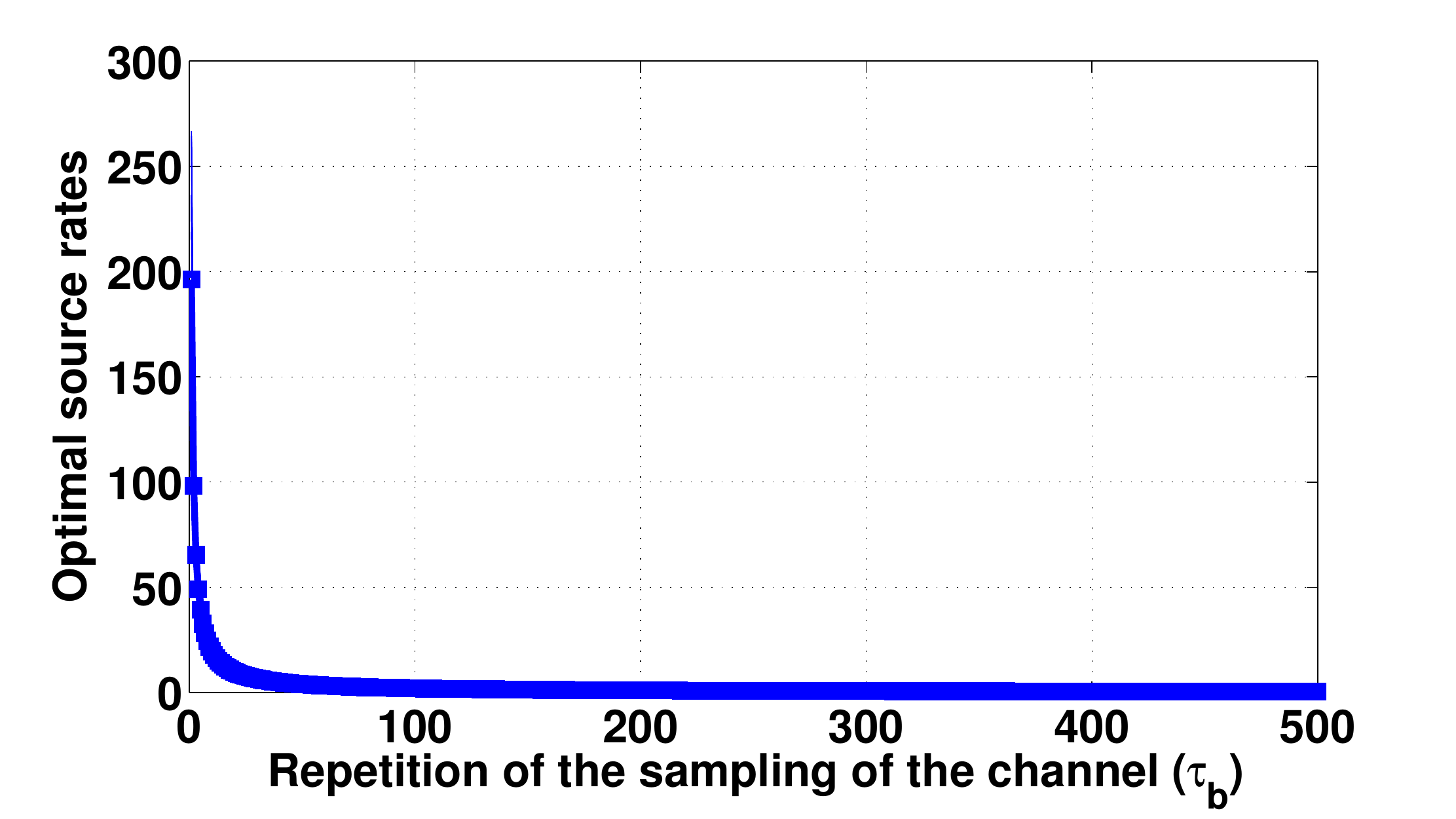}
    \label{fig:RoutingCongestionTimeVarying}}
 \subfigure[Impact of $n$ on the time to convergence (iterations) for time varying utilities.]
      {\includegraphics[width=2in]{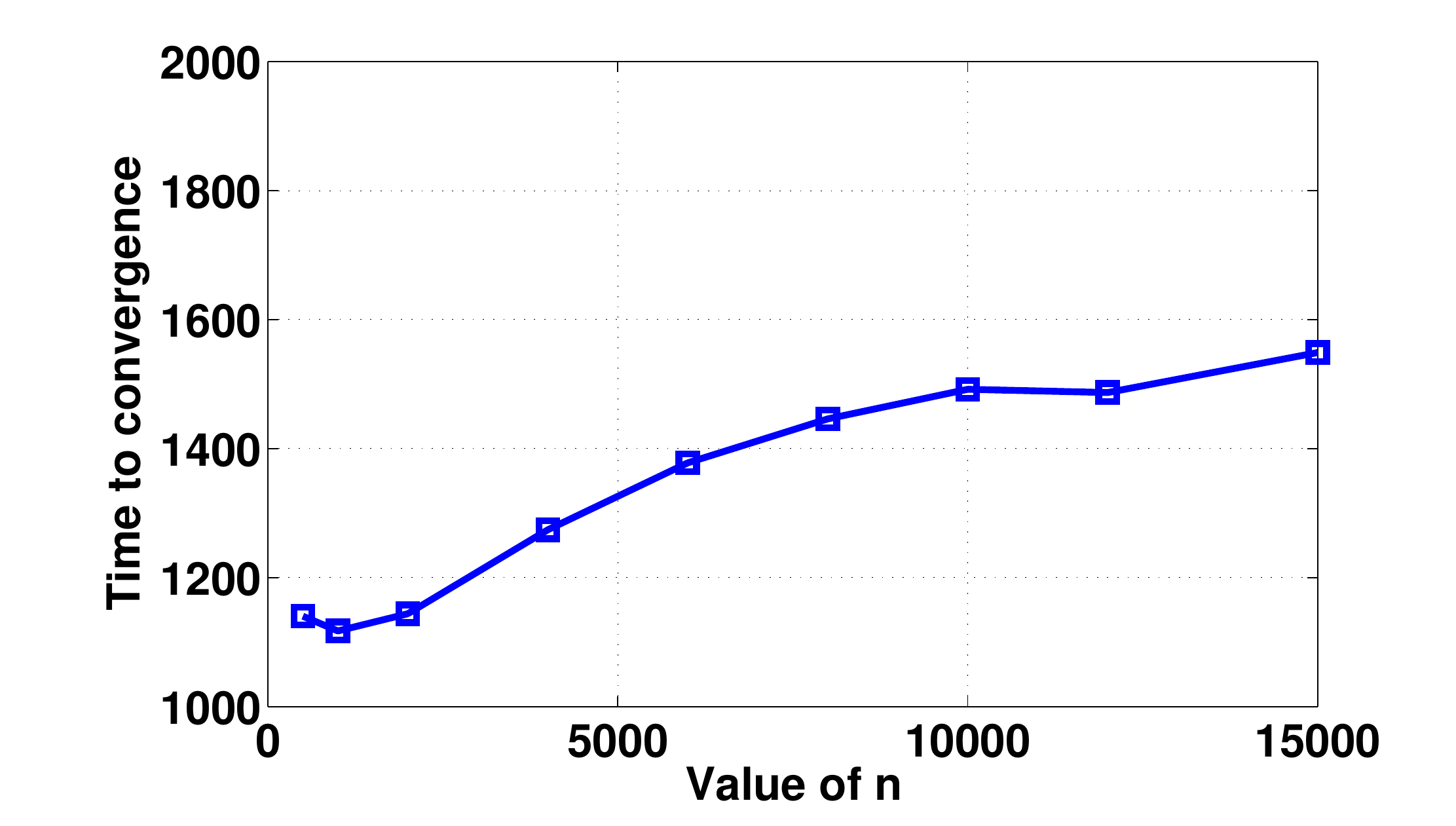}
    \label{fig:ntimeconv1}}
        \subfigure[Impact of $n$ on the time to convergence (iterations) for time invariant utilities.]
      {\includegraphics[width=2in]{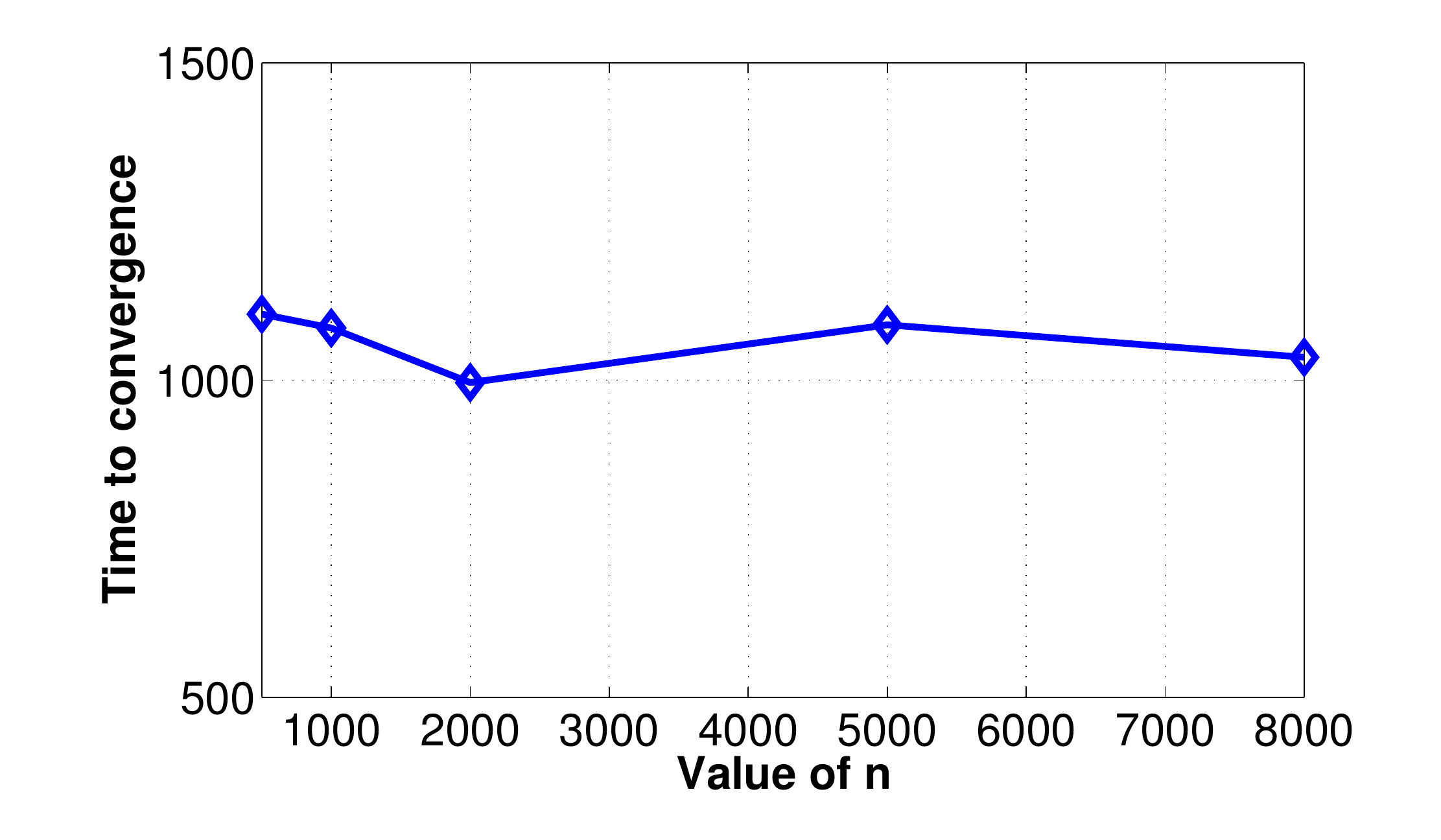}
    \label{fig:ntimeconv2}}
    \caption{Routing \& Congestion control: Time varying utilities \& Impact of $n$ on convergence.}
\end{figure}

\section{Conclusions}
\label{sec:conclusions} In this paper we presented, analyzed and
evaluated a framework of NUM for performing routing, scheduling,
congestion and power control under stochastic possibly
non-stationary LTF or STF wireless channels modeled by SDEs. The
continuous stochastic non-stationary wireless channels along with
the consideration of transient phenomena lead to a problem
formulation that can also tackle non-convex and time-varying
objective functions in an optimal way. Power control aims at
increasing users' experience allowing for higher source rates while
also improving the energy efficiency. In the case of LTF, we prove
that higher values of the diffusion coefficient of the power loss
lead to higher optimal users' utilities, a fact that cannot be
captured by the conventional NUM problem's formulation. Numerical
results evaluate the latter along with the convergence properties of
our proposed algorithms and the effect of diverse parameters on it.
The efficiency of power control and the conditions under which an
online implementation of the proposed approach is possible are also
investigated. Finally, our proposed NUM-based framework may
constitute a core for devising efficient cross-layer algorithms for
the network operation that incorporate transient or non-stationary
phenomena.

\begin{small}
\section{Acknowledgment} This research is co-financed by the European
Union (European Social Fund) and Hellenic national funds through the
Operational Program 'Education and Lifelong Learning' (NSRF
2007-2013). (under ``ARISTEIA" 1260). M.L. acknowledges support from
the NSRF Research Funding Program Thales: Optimal Management of
Dynamical Systems of the Economy and the Environment MIS375586.
\end{small}

\appendices

\bibliographystyle{plain}
\bibliography{references}

\end{document}